\documentclass[12pt]{article}

\usepackage{amsmath}
\usepackage{amssymb}
\usepackage{amsthm}
\usepackage{graphicx}
\usepackage{caption}
\usepackage{fourier} 
\usepackage{soul} 
\usepackage{paralist}
\usepackage[round]{natbib}   

\usepackage{tikz}
\tikzset{
	treenode/.style = {shape=rectangle, rounded corners,
		draw, align=center,
		top color=white, bottom color=blue!20},
	root/.style     = {treenode, font=\large, bottom color=red!30},
	env/.style      = {treenode, font=\ttfamily\normalsize},
	dummy/.style    = {circle,draw},
	ancester/.style = {shape=rectangle, rounded corners,
		draw, align=center,
		top color=white, bottom color=red!40}
}

\usepackage{pgfplots}
\pgfplotsset{compat=newest}
\usepgfplotslibrary{groupplots}
\usepgfplotslibrary{dateplot}

\usepackage{color}

\usepackage{todonotes}
\presetkeys{todonotes}{size=\scriptsize}{}

\usepackage[margin=3cm]{geometry}

\usepackage[T1]{fontenc} 

\usepackage{nicefrac}

\usepackage{cancel} 
\usepackage{hyperref}

\newtheorem{theorem}{Theorem}
\newtheorem{proposition}{Proposition}
\newtheorem{corollary}{Corollary}
\newtheorem{lemma}{Lemma}

\theoremstyle{definition}

\theoremstyle{remark}
\newtheorem{remark}{Remark}

\def\E{\mathbb{E}}
\def\P{\mathbb{P}}
\def\R{\mathbb{R}}

\def\d{\partial}
\def\Z{\mathbb{Z}}

\def\cF{{\cal F}}
\def\cG{{\cal G}}

\def\cB{{\cal B}}

\newcommand{\dd}{\mathrm{d}}

\newcommand{\Sbb}{\mathbb{S}}
\newcommand{\M}{{\mathcal{M}_P\left(\Sbb\times \mathbb{R}_+\right)}}
\newcommand{\X}{{\mathbb{X}}}

\newcommand{\PJ}{{\mathrm{P}}}
\newcommand{\ind}{{\mathbf{1}}}
\newcommand{\bZ}{{\mathbf{Z}}}

\newcommand{\UC}{{\mathcal{UC}_b^0}}

\newcommand{\uls}[1]{\underline{\smash{#1}}}

\begin{document}

\title{A branching model for intergenerational telomere length dynamics}
\author{Athanasios Benetos$^{1,2}$, Olivier Coudray, Anne G\'egout-Petit$^{3}$, Lionel Len\^otre$^{4,5,6}$\\
     Simon Toupance$^{1}$, 
and Denis Villemonais$^{3,7,8}$}

\footnotetext[1]{Université de Lorraine, Inserm, DCAC, F-54000, Nancy, France}
\footnotetext[2]{Université de Lorraine, CHRU-Nancy, Pôle "Maladies du Vieillissement, Gérontologie et Soins Palliatifs", F-54000, Nancy, France}
\footnotetext[3]{Université de Lorraine, CNRS, Inria, IECL, F-54000 Nancy, France}
\footnotetext[4]{Université de Haute-Alsace, IRIMAS UR 7499, F-68200 Mulhouse, France}
\footnotetext[5]{Université de Haute-Alsace, UMR 7044 Archimède, F-67000 Strasbourg, France}
\footnotetext[6]{Inria, PASTA, F-54000, Nancy, France }
\footnotetext[7]{Institut universitaire de France (IUF)}
\footnotetext[8]{Corresponding author, email: denis.villemonais@univ-lorraine.fr}

\maketitle

\begin{abstract}
We build and study an individual based model of the telomere length's evolution in a population across multiple generations. This model is a continuous time typed branching process, where the type of an individual includes its gamete mean telomere length and its age. We study its Malthusian's behaviour and provide numerical simulations to understand the influence of biologically relevant parameters.
\end{abstract}

\textit{Keywords: } Telomeres dynamics; Population dynamics; Aged structured model; Branching processes; Quasi-stationary distributions.

\textit{MSC2020 Classification: } 60K40; 60J80; 60J85; 60F99.


\section{Introduction}

Telomeres are specialized nucleoproteic structures that form protective caps at each end of eukaryotic chromosomes. They consist of non-coding repetitive nucleotide sequences associated with a family of proteins. These structures maintain genomic integrity through their capacity to prevent end-to-end chromosome fusions and chromosome extremity recognition as DNA breaks. With each cell division, part of the DNA located at telomeres’ end is lost due to incomplete replication, a phenomenon known as the “end replication problem”. Therefore, this leads to progressive telomere shortening in somatic cells, and in the end to critically short telomeres, which triggers replicative senescence, a state in which cells cease to divide~\cite{XuDucEtAl2013}. We refer the reader to : \cite{EntringerPunderEtAl2018} for an account on telomeres and on their length's dynamics with respect to the age of individuals, \cite{WhittemoreVeraEtAl2019} for a study of the relation between telomere length and life span across different species, and~\cite{LaberthonniereMagdinierEtAl2019} for a survey on the effect of telomere length on individuals health. In humans, it is acknowledged that short telomere lengths are determinants in the development of age-related diseases such as atherosclerosis~\cite{BenetosToupanceEtAl2018}. Telomere lengths also have a strong impact on the lifespan of an individual, but the statistical link remains unclear~\cite{GleiGoldmanEtAl2016}. 

Somatic cells are dysfunction is implicated in a large number of diseases that are suspected to arise from genomic instability or senescence, and thus potentially linked with telomere length. These cells show different phases of telomere length shortening~\cite{FrenckBlackburnEtAl1998}.  From embryonic phase up to childhood, the mean telomere length decreases strongly, while it decreases slowly in adulthood. The erosion speed is similar among adults, but the first phase is individual dependent and is influenced by many environmental factors (for instance intrauterine stress exposure~\cite{EntringerEpelEtAl2011}, childhood obesity~\cite{BuxtonWaltersEtAl2011}, exposure to violence during childhood~\cite{ShalevMoffittEtAl2013}). But all these mechanisms reduce a starting length that is inherited from parents, and consequently from previous generations. Understanding the transmission of telomere lengths across generations within a population is therefore essential.

Telomere length is a highly heritable trait \cite{HjelmborgEtAl2015,BroerEtAl2013,HonigEtAl2015} and is partially influenced by genetic factors. The telomere length of a child is strongly related to gametes' telomere lengths of the parents, particularly of the father~\cite{AvivSusser2013, DeMeyerRietzschelEtAl2007,NordfjaellLarefalkEtAl2005}. Telomere lengths dynamics of male gametes are very different from those of somatic cells since they are subject to the activity of telomerase, an enzyme responsible for maintenance of the length of telomeres~\cite{ZverevaShcherbakovaEtAl2010}. It results a tendency of telomere lengths in male gametes to increase with age~\cite{AvivSusser2013}. The birth-rate as a function of age in a population is thus expected to have an influence over the evolution of telomere lengths distribution within a population.  Therefore, knowing that parents in many countries are having children at an older age than half a century ago (see Figure~\ref{fig:birthrate}), one might expect children to have longer telomeres on average. As explained in~\cite{AvivSusser2013}, higher paternal age at conception has well-documented detrimental effects; these could be offset by beneficial effects due to telomere lengthening induced by paternity at later age.  At the same time, the average length of telomeres in the population changes over relatively short time scales. A striking consequence of this fact is the difference in telomere shortening with age measured in longitudinal versus cross-sectional studies~\cite{HolohanDeMeyerEtAl2015} and with potential implications for public health. This blurs the impact of heredity and prompts the development of models to better understand its real effect.

We propose a probabilistic process that models the evolution through generations of the size of a population as well as the average length of the telomeres of its individuals (see e.g.~\cite{BourgeronXuEtAl2015a,LeeKimmel2020,MattarocciBerardiEtAl2021a,OlofssonKimmel1999a} for models of telomere length's dynamic at the microscopic level). Each individual carries the telomere length of its gametes at breeding age, this does not detract from the generality because it would be possible to obtain the average telomere length of an individual's somatic cells at any age by applying a transfer function obtained by regression to the telomere length. The individuals are asexual; we can imagine that they are a reproductive couple of humans. This is a first model, and we do not want to introduce too much complexity. Age is the second characteristic of an individual since the length of the telomeres of the gametes depends on it. Individuals reproduce during a given period,  in the context of a human couple it emulates the time between the formation of a couple, which is a sort of breeding age, and the menopause of the woman, and at a certain rate depending on age but not on the telomere length of gametes. Individuals also reproduce  independently. Finally, the length of the telomeres at puberty is given by a transition function taking into account age and simulating the action of telomerase on the telomeres of gametes. We will specify its choice later.

\begin{figure}[!htb]
    \center{\includegraphics[height=9cm]{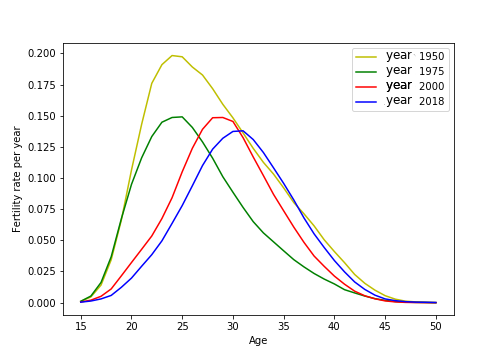}}
    \caption{\label{fig:birthrate} Birth-rate as a function of age in 1950, 1975, 2000 and 2018 (INSEE)}
\end{figure}

Mathematically speaking, the model is a Crump Mode Jagers typed branching process with age (abreviated CMJ) which is a generalization of the age dependent branching processes described in~\cite[Chapter~VI]{Harris1963}. We refer the reader to~\cite{JagersNerman1984} (age structured branching processes without types) and~\cite{Jagers1989} (age structured branching processes with types) for an introduction; \cite{Olofsson2009} provides refined convergence results and~\cite{Bertoin2017} a construction in a growth fragmentation setting. In these branching models, one considers the genealogy of the population, each individual being marked by its type (say $s_i$ for individual $i$) and its birth time (say $t_i$), as represented in Figure~\ref{fig:exCMJ1}. The age of an individual is denoted by $a\in[0,+\infty)$ and evolves linearly in time; the mean telomere length of its gametes, abbreviated by GTL, is designated by $l\in[0,+\infty]$. The breeding age $a_p>0$ is assumed to be fixed across the population. The birth rate is a function of age $b:[0,+\infty)\mapsto[0,+\infty)$ satisfying $b(a)=0$ for $a<a_p$ (see Figure~\ref{fig:birthrate}). With these notations, the dynamics describe right above reads as follows. Each alive individual in the population gives birth to one new individual at random times, independently from each others and from their GTL at puberty at rate $b(a)$ at age $a$. The GTL of an individual at age $a\geq a_p$ grows linearly with time, with a fixed slope $\alpha>0$, so that it is given by $(a-a_p)\alpha$. When a newborn appears in the population, its initial age is $0$ and its GTL at puberty parameter is chosen randomly, depending on the GTL of its parent at the time of birth, denoted by $GTL_b$. It follows a truncated Gaussian distribution with mean $GTL_b-\mu$, $\mu>0$ being the mean erosion of telomeres during the pregnancy/childhood phase, and variance $\sigma^2>0$; truncation occurs in the interval $[l_{min},l_{max}]$, where $l_{min}>0$ and $l_{max}>0$ are respectively the minimal and maximal length of any individual. Finally, we suppose that each individual dies at a same age $a_d>a_p$ (this last assumption could be weakened, at the expense of additional technicalities).
\begin{figure}[h]
    \center{\includegraphics[height=5cm]{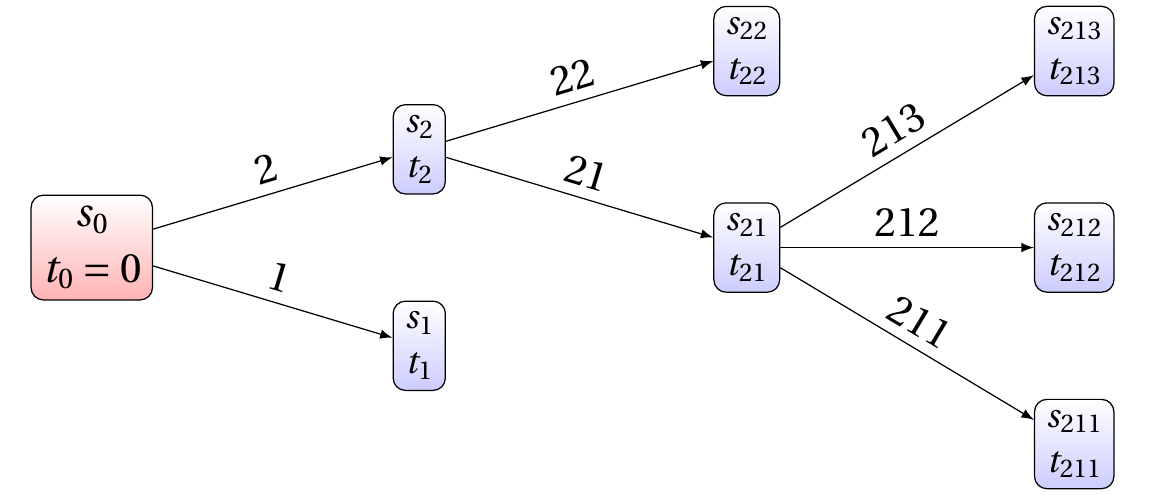}}
    \caption{\label{fig:exCMJ1} Example of a CMJ branching process.}
\end{figure}
Of course, one can represent the process by unfolding the genealogy along the time dimension, as done in Figure~\ref{fig:exCMJ2}. 
\begin{figure}[h]
      \center{\includegraphics[height=5cm]{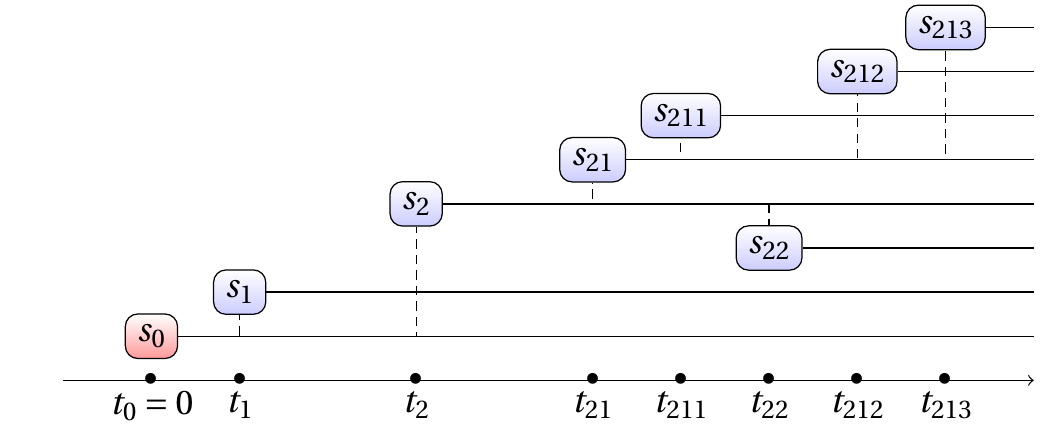}}
    \caption{\label{fig:exCMJ2} Same CMJ process but unfolded along a time axis.}
\end{figure}

 We provide theoretical results and illustrative numerical experimentations. 
 The first main theoretical result is a many-to-one formula for CMJ branching processes, where we show that the potential kernel of those models can be represented via two different Feynman-Kac formulas whose equivalence relies on a simple and classical, yet powerful, argument : their associated semi-groups have the same infinitesimal generators~\cite{Pazy1983}. The second relevant result concerns some spectral properties of the non-conservative semi-group involved : it converges exponentially fast to a ground state, thanks to some recent results on quasi-stationary distributions~\cite{ChampagnatVillemonais2016b}. At the end, it gives a representation of the Malthusian parameter of the model. These theoretical results   are independent of the choices of the parameters $b$, $\alpha$, $\mu$, $\sigma$, $l_{min}$, $l_{max}$ and $a_d$, and they apply to more general CMJ branching processes.

In numerical simulations, the birth-rate function $b$ is chosen according to Figure~\ref{fig:birthrate}, $\alpha=0.017$ (a value justified in \cite{AvivSusser2013}), $l_{min}=5$ kb and $l_{max}=25$ kbp (the quantity $1$ kbp corresponds to a length of DNA of one thousand base pairs), and $\sigma=0.1$ (this is chosen so that the standard deviation in the whole population is of the same order as in~\cite{DeMeyerRietzschelEtAl2007} for the telomere length distribution in their cohort of male adults), although different measurements may be relevant. Different values of $\mu$ and shifted birth rate functions (corresponding to shifts in the parental reproduction age) are considered, and we observe qualitatively the changes implied by these perturbations on the equilibrium (long time) distribution of telomeres in the population, and on the relations between telomere length distribution, father age at birth and parental birth year.

In Section~\ref{sec:modeldef}, we introduce the CMJ branching process that models the behavior of the GTL in a population. Then we state in Section~\ref{sec:poissonian} a Feynman-Kac representation (via a many-to-one formula) of the potential kernel of our branching process under a Poisson branching time assumption. Another many-to-one formula related to a different Feynman-Kac representation of the potential kernel is stated in Section~\ref{sec:nonpoissonian}. The exponential convergence of these non-conservative semi-groups is given in Section~\ref{sec:Malthus}. Finally, we present numerical simulations to illustrate the effect of changes in $\mu$ and of the right shift of the birth-rate curves on the GTL distribution in a population (see Section~\ref{sec:sim}).

\paragraph{Notations :} $\R_+$   denotes the set of non-negative real numbers, $\mathcal M_P(E)$ the set of finite discrete measures on $E$,  and $\|\cdot\|_{TV}$ the total variation distance between measures. As usual, for a mathematical object $x$ belonging to a set $X$, $\delta_x$ stands for the Dirac mass at $x$.

\section{Definition of the model and many-to-one formulas}

\subsection{Definition of the model}
\label{sec:modeldef}
We define an age-dependent branching process with a type belonging to a Polish space $\Sbb$. Each individual is represented by an atom $\delta_{s,t}$, where $s\in \Sbb$ is the type of the individual
and $t\in\R_+$ is the birth date of the individual. The $n^{th}$ generation is  a finite discrete measure on $\Sbb\times\R_+$,  denoted by $X_n\in\M$. 

\begin{remark}
	In the introduction and in our simulations section,  $\Sbb=[l_{\min},l_{\max}]$. However, it may be desirable to include additional traits in the type space, that may be transmitted from parents to childrens or shared among brotherhood, for instance the social environment, childhood exposition to violence, ethnicity or genetic diseases.
\end{remark}

\medskip

Let $b:\Sbb\times\R_+\to\R_+$ be a measurable, compactly supported and bounded function and $\gamma$ a continuous probability kernel from $\Sbb\times\R_+$ to $\Sbb$. In our model, $b(s_0,a)$ represents the reproduction rate for an individual with type $s_0$ and age $a$, and $\gamma_{s_0,a}(\dd s)$ is the type's law of a child born from a father with type $s_0$ and  age $a$. Said differently, we denote by $\PJ_{s_0}$ the law of a Poisson point process in $\Sbb\times \R_+$, with intensity $b(s_0,a)\gamma_{s_0,a}(\dd s)\,\dd a$, and assume that the progeny's distribution of an individual with type $s_0\in\Sbb$ at time $0$ is given by $\PJ_{s_0}$.

The branching process is constructed recursively, generation after generation.  Let $X_0=\delta_{s_0,0}$ be a fixed punctual measure representing the original state of the population at time $0$, constituted of one individual with type $s_0$ and birth date $0$. Assuming that $X_n=\sum_{i=1}^{\overline{X_n}}\delta_{s_i^n,t_i^n}$, where $\overline{X_n}:=X_n(\Sbb\times\R_+)$ is the number of individuals in generation $n$, we define
\[
X_{n+1}=\sum_{i=1}^{\overline{X_n}} \theta_{t_i^n}\circ \xi_{s_i^n}^{n+1},
\]
where the $\xi_{s_i^n}^{n+1}$, $1\leq i\leq \overline{X_n},i<\infty$, are random independent discrete measures with respective laws $\PJ_{s_i^n}$, and where, for all $(s_1,t_1),\ldots,(s_k,t_k)\in\Sbb$ and all $t\in\R_+$,
\[
	\theta_t\circ\sum_{i=1}^k \delta_{s_i,t_i}:=\sum_{i=1}^k \delta_{s_i,t_i+t}.
\]
Informally, $X_n(A\times B)$ should be interpreted as the number of individuals of the $n^{th}$ generation, with type in $A$  and with birth date in $B$.

We emphasize that, since $b$ has compact support and is bounded, each random measure $X_n$ can be written under the form
\[
X_n=\sum_{i=1}^{\overline{X_n}}\delta_{s_i^n,t_i^n},
\]
where $\overline{X_n}<+\infty$ and $t_1^n<t_2^n<\cdots$. Given $s_0\in\Sbb$, we denote by $\P_{(s_0,t_0)}$ the law of $(X_n)_{n\in\Z_+}$ when $X_0=\delta_{(s_0,t_0)}$ almost surely, and by $\E_{(s_0,t_0)}$ the corresponding expectation.

%
%
%

  Following~\cite[Section~5]{Jagers1989}, we define the \textit{reproduction kernel $\mu$} from $\Sbb\times \R_+$ to $\Sbb\times \R_+$ as
\[
\mu(s,A\times B)=\int_{\M} \xi(A\times B)\,\PJ_s(d\xi),\quad s\in\Sbb,A\in \cB(\Sbb), B\in \cB(\R_+).
\]
Hence,  given an individual with type $s\in\Sbb$ at time $0$, the quantity $\mu(s,A\times B)$ gives the mean number of its children whose type are in $A$ and whose birth date is in $B$. We also define the iterates of $\mu$ as $\mu^0(s,\cdot)=\delta_{(s,0)}$ and, by iteration,
\[
\mu^{n+1}(s,A\times B)=\int_{\Sbb\times \R_+} \mu(r,A\times(B-u))\,\mu^n(s,\dd r\times \dd u).
\]

Since this is not stressed out in~\cite{Jagers1989}, we give a short proposition giving the meaning of $\mu^n$ in terms of the composition $X_n$  of the population at generation $n$: $\mu^n(s,A\times B)$ gives the mean number of individuals of the $n^{th}$ generation whose type is in $A$ and whose birth date is in $B$.

\begin{proposition}
	For all $n\in\Z_+$, all $s_0\in\Sbb$ and all measurable sets $A\subset \Sbb$ and $B\subset \R_+$, we have
	\[
	\E_{(s_0,0)}\left[X_n(A\times B)\right]=\mu^n(s_0,A\times B).
	\]
\end{proposition}

\begin{proof}
	We show this result by iteration over $n$. The cas $n=0$ is immediate. Assume now that the property holds true for $n\in\Z_+$. Then, by definition of $X_{n+1}$,
	\begin{align*}
	\E\left(X_{n+1}(A\times B) \mid X_n = \sum_{i=1}^{\overline{X_n}} \delta_{s_i^n, t_i^n}\right)
	&= \sum_{i=1}^{\overline{X_n}} \int_{\M}\theta_{t_i^n}\circ\xi(A\times B) \PJ_{s_i^n}(d\xi).
	\end{align*}
	Taking the expectation and using the induction assumption, we obtain
	\begin{align*}
	\E\left(X_{n+1}(A\times B)\right)&=\int_{\Sbb\times\R_+} \int_{\M}\theta_{u}\circ \xi(A\times B) \PJ_{r}(d\xi) \mu^n(s_0,\dd r\times \dd u)\\
	&=\int_{\Sbb\times\R_+} \int_{\M}\xi(A\times (B-u)) \PJ_{r}(d\xi) \mu^n(s_0,\dd r\times \dd u)\\
	&=\int_{\Sbb\times\R_+}\mu(r,A\times(B-u)) \mu^n(s_0,\dd r\times \dd u)\\
	&=\mu^{n+1}(s,A\times B).
	\end{align*}
\end{proof}

Similarly, for any $\lambda\in\R$, we define as in~\cite{JagersNerman1984} the kernel $\mu_\lambda$ as
\[
\mu_\lambda(r,\dd s\times\dd u)=e^{-\lambda u}\,\mu(r,\dd s\times\dd u),
\]
$\mu_\lambda^0(s,\cdot)=\delta_{(s,0)}$ and, iteratively,
\[
\mu_\lambda^{n+1}(s,A\times B)=\int_{\Sbb\times\R_+} \mu_\lambda(r,A\times(B-u))\,\mu_\lambda^n(s,\dd r\times \dd u).
\]
The proof of the following result is similar to the previous one and is thus left to the reader.

\begin{proposition}
For all $\lambda\in\R$, $n\in\Z_+$, all $s_0\in\Sbb$ and all measurable sets $A\subset \Sbb$ and $B\subset \R_+$, we have
\[
\E_{(s_0,0)}\left[\sum_{i=1}^{\overline{X_n}}\delta_{s_i^n,t_i^n}e^{-\lambda t_i^n}\right]=\mu_\lambda^n(s_0,A\times B).
\]	
\end{proposition}

\subsection{Many-to-one formula for Poissonian reproduction times}

\label{sec:poissonian}

The aim of this section is to provide a first many-to-one formula, which allows to express the potential kernel (which involves expectations over \textit{many} individuals) as a Feynman-Kac type expression (which only involves \textit{one} trajectory).

In order to do so, we consider the piecewise-deterministic Markov process $(Z_t)_{t\in[0,+\infty)}$ with values in $\Sbb\times \R_+$, which evolves according to the flow $((s_0,a),t)\mapsto (s_0,a+t)$ and, at a rate $b(Z_t)$, jumps according to the probability measure $\gamma_{Z_t}(\dd s)\otimes \delta_0(\dd a)$. We refer the reader to~\cite{Davis1984,Davis1993,AzaisBardetEtAl2014} for general aspects of the theory of piecewise-deterministic Markov processes. We denote respectively by $Z^{(s)}_t$ and by $Z^{(a)}_t$ the first and second component of $Z_t$, respectively in $\Sbb$ and $\R_+$. The component $Z^{(a)}_t$ should be interpreted as the age of an individual (since the last jump), so that $t-Z^{(a)}_t$ is the birth date of the individual (that is the last jump time before time $t$). In what follows, the law of $Z$ with initial position $(s,t)\in\Sbb\times \R_+$ is denoted by $\P_{s,t}^Z$ and its associated expectation $\E_{s,t}^Z$.

Following~\cite{Jagers1989}, we consider the potential kernel $\nu$, defined by
\[
\nu(s,A\times B)=\sum_{n=0}^\infty \mu^n(s,A\times B),
\]
for all $s\in\Sbb$ and all measurable subsets $A\subset\Sbb$ and $B\subset\R_+$. The following many-to-one formula is the main result of this section.

\begin{proposition}
	\label{prop:Zmanytoone}
	We have, for all  $t \geq 0$, all $s_0 \in \Sbb$ and all measurable sets $A\subset \Sbb$ and $B\subset \R_+$, 
	\[
	\nu\big(s_0,A\times (B\cap[0,t])\big)=\E^Z_{s_0,0}\left[\ind_{Z^{(s)}_t\in A,\ t-Z^{(a)}_t\in B}\ \exp\left(\int_0^t b\big(Z_u\big)\,\dd u\right)\right].
	\]
\end{proposition}

The proof of this proposition is postponed to Section~\ref{sec:proof3}. Our proof's strategy is first to represent $\nu$ as the expectation with respect to a process exploring a random branch of the model (see next section), and second to prove that both representations coincide.

We also define, for all $\lambda\in\R$, $\nu_\lambda=\sum_{n=0}^\infty \mu_\lambda^n$ and obtain the following corollary.
\begin{corollary}
	\label{cor:useful}
	We have, for all  $\lambda\in\R$, all $t \geq 0$ and all $s_0 \in \Sbb$,
	\[
	\nu_\lambda(s_0,\Sbb\times[0,t])=e^{-\lambda t}\E_{s_0,0}^Z\left(e^{\lambda Z^{(a)}_t}\,e^{\int_0^t b(Z_u)\,\dd u}\right)
	\]

\end{corollary}


\subsection{Random exploration of a CMJ branch}
\label{sec:nonpoissonian}

We describe now a continuous time random process $(Y_t)_{t\in[0,+\infty}$ which explores randomly a branch of the generation tree of a CMJ branching process. This process takes values in $\X:=\Sbb\times \R_+\times \M$, the first component corresponding to the type of the current individual, the second component is the age of the current individual, and the third one is the total progeny of the current individual (starting at its birth time). Informally, the process $Y$ starts at the ancestor position (which includes its type, its birth time, and its progeny) at time $0$ and stays idle up to the first reproduction time. Then it either jumps on the position of the newborn (with probability $\nicefrac12$), or it remains at its position (with probability $\nicefrac{1}{2}$). Then it stays at the same position up to the next reproduction time and so on (see Figure~\ref{fig:exCMJ3}).

\begin{figure}[h]
    \center{\includegraphics[height=5cm]{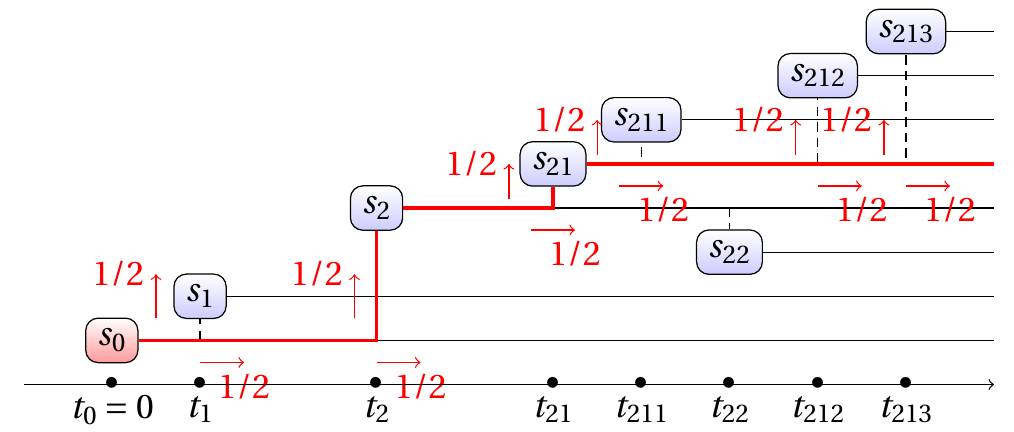}}
	\caption{\label{fig:exCMJ3} Example of a trajectory of the process $Y$: at each reproduction time, the process chooses one of the branches with probability $\nicefrac12$.}
\end{figure}

Let us now define more formally the process $Y$. Let $\tau:\X \rightarrow \R_+$ and $\Lambda:\X\to \Sbb$ be defined, for all $(s,a,\xi) \in \X$, by 
\[ 
\tau(s,a,\xi) = \inf\left\lbrace t > 0,\ \xi(\Sbb\times(a,a+t]) = 1 \right\rbrace
\]
and
\[
\Lambda(s,a,\xi) = l\text{ such that }\xi\{(l,\tau(s,a,\xi))\}=1,
\]
with the convention $\inf\  \emptyset  = +\infty$ and, if $\tau(s,a,\xi)=+\infty$, $\Lambda(s,a,\xi)=s$. Informally, $\tau(s,a,\xi)$ gives the first reproduction time after time $a$, while $\Lambda(s,a,\xi)$ gives the type of the then born child. Now let $\Pi_0$ and $\Pi_1$ be two probability kernels defined by
\begin{align*}
\Pi_0\ :\ &\X\times\mathcal{X} \rightarrow [0,1]\\
& \left((s,a,\xi), A\times B \times C\right) \mapsto \delta_{(s,a+\tau(s,a,\xi),\xi)}(A\times B \times C)
\end{align*}
and
\begin{align*}
\Pi_1\ :\ &\X\times\mathcal{X} \rightarrow [0,1]\\
& \left((s,a,\xi), A\times B \times C\right) \mapsto \delta_{\Lambda(s,a,\xi),0}(A\times B)\ \PJ_{\Lambda(s,a,\xi)}(C),
\end{align*}
where $\cal X$ is the product $\sigma$-field on $\X$ and we recall that $\PJ_s$ is the law of the progeny of an individual with type $s\in \Sbb$.
On the one hand, given $(s,a,\xi)$, $\Pi_0(s,a,\xi)$ is a Dirac measure at $(s,a+\tau(s,a,\xi),\xi)$ and will be used as the jump kernel in the event where $Y$ remains on the father's branch at the reproduction time $a+\tau(s,a,\xi)$. On the other hand, $\Pi_1(s,a,\xi)$ will be used as the jump kernel when $Y$ jumps on the branch of the new-born, since $\Lambda(s,a,\xi)$ is the type of the child, $0$ is its age at the time of birth, and $\PJ_{\Lambda(s,a,\xi)}$ is the law of its progeny.

We first define the included chain of $Y$ at jump times, denoted by $(\sigma_k,W_k)_{k\geq 0}$, where $\sigma_k$ denotes the $k^{th}$ jump time, while $W_k$ denotes the $k^{th}$ position after the jump (note that the size of a jump might be $0$ if the process remains on the father's branch). Let $(\sigma_0,W_0)\in\R_+\times\X$ and define the process iteratively as follows. Given $(\sigma_k,W_k)$, $k\geq 0$,
\begin{itemize}
	\item we set $\sigma_{k+1}=\sigma_k+\tau(W_k)$,
	\item we choose $W_{k+1}$ according to $\Pi_{\varepsilon_{k+1}}(W_k,\cdot)$, where $\varepsilon_{k+1}$ is a Bernoulli random variable with parameter $\nicefrac12$ independent from the rest of the process.
\end{itemize}
If $\sigma_{k+1}=+\infty$, then the Markov chain is stopped. We then formally define $(Y_t)_{t\in[0,+\infty)}$ as
\[
Y_t = \sum_{k = 0}^{+ \infty} \ind_{\sigma_k \leq t < \sigma_{k+1}} \vartheta_{t-\sigma_k}W_k,
\]
where, for all $u\in\R_+$ and $(s,a,\xi)\in\X$, $\vartheta_u (s,a,\xi)=(s,a+u,\xi)$. Note that the first component of $Y_t$ represents the current individual's type, the second component its age at time $t$, and the third component its progeny.

We also define the process $(N_t)_{t\geq 0}$ counting the number of jumps
\[
N_t = \sum_{k=1}^{+\infty} \ind_{\sigma_k \leq t}.
\]
Note that, since $b$ is assumed to be bounded, we have
\[
\sigma_n\xrightarrow[n\to+\infty]{}+\infty\quad\text{almost surely}.
\]

We define the filtration $(\cG_t)_{t\geq 0}$ by
\[
\mathcal{G}_t = \sigma\big((\sigma_k \leq t < \sigma_{k+1})\cap A,\ \text{where }k \geq 0\text{ and }A \in \sigma(W_0,W_1,\ldots,W_k) \big).
\]

\begin{proposition}	
	\label{prop:Ysemigroup}
	The process $(N,Y)$ is a Markov process with respect to the filtration $\cG$. More precisely, for all $f \in L^\infty\left(\mathbb Z_+\times \X \right)$ and all $t,u \geq 0$,
	\[
	\E\left(f(N_{t+u}, Y_{t+u})\ |\ \mathcal{G}_t\right) = Q_u f(N_t, Y_t)\]
	where $Q$ is the semi-group\footnote{by a semi-group, we mean that $Q_{s+t}=Q_sQ_t$ for all $s,t\geq 0$.} associated to the process $(N,Y)$, and is equal to
	\[
	Q_u f(k,y) =  \ind_{u < \tau(Y_t)}\ f(N_t, \vartheta_u Y_t) +  \sum_{j=0}^{+\infty} \int_{\X}\,\Pi(Y_t,\dd w)\, \E^{\sigma,W}_{\tau(Y_t),w}\left(\ind_{\sigma_j\leq t+u<\sigma_{j+1}}\ f(N_t+1+j,\vartheta_{u-\sigma_j}W_j)\right),
	\]
	with $\Pi=\frac{1}{2}\Pi_0+\frac{1}{2}\Pi_1$ and where $\E_{\sigma_0,W_0}^{\sigma,W}$ denotes the expectation with respect to the law of $(\sigma_k,W_k)_k$ starting from $(\sigma_0,W_0)$.
\end{proposition}
The proof of Proposition~\ref{prop:Ysemigroup} is detailed in Section~\ref{sec:proof1}. We emphasize that it does not make direct use of the Poissonian nature of the jump mechanism.


%
%

In the construction of $Y$, and more precisely for the construction of $(\sigma_k,W_k)$, we use a sequence of independent Bernoulli random variables  $(\varepsilon_k)_{k \geq 0}$, which encodes the choice of $Y$ to remain on the father's branch ($\varepsilon_k=0$) or to continue on the child's branch ($\varepsilon_k=1$) at each time $\sigma_k$. 
The following result gives a representation of $\mu^n$ in terms of $(N,Y)$. In the following proposition, $Y^{(s)}_t$ and $Y^{(a)}_t$ denote respectively the first and second component of $Y_t$, with value in $\Sbb$ and $\R_+$ respectively. 
In particular, since $Y^{(a)}_t$ should be interpreted as the age of the individual chosen by $Y$ at time $t$, the quantity $t-Y^{(a)}_t$ corresponds to its birth date.

\begin{proposition}
	\label{prop:Ymanytoone}
	For all measurable sets $A\subset \Sbb$ and $B\subset  \R_+$ and all $n \in \mathbb Z_+$, we have, for all  $t \geq 0$ and all $s_0 \in \Sbb$, 
	\[
	\mu^n(s_0,A\times (B\cap[0,t]))=\int_{\M} \,\PJ_{s_0}(\dd \xi)\E^Y_{s_0,0,\xi}\left(\ind_{Y^{(s)}_t\in A,t-Y^{(a)}_t\in B}\ 2^{N_t}\   \ind_{\sum_{i = 1}^{N_t}\varepsilon_i = n}\right),
	\]	
	where $\E^Y_{s_0,0}$ denotes the expectation with respect to te law of $Y$ when the law of $Y_0$ is $\delta_{s_0}\otimes\delta_0\otimes\P^J_{s_0}$.
	In particular,
	\[
	\nu(s_0,A\times B)=\int_{\M} \,\PJ_{s_0}(\dd \xi)\E^Y_{s_0,0,\xi}\left(\ind_{Y^{(s)}_t\in A,t-Y^{(a)}_t\in B}\ 2^{N_t}\right).
	\]
\end{proposition}


The proof of Proposition~\ref{prop:Ymanytoone} is detailed in Section~\ref{sec:proof2}. As the proof of Proposition~\ref{prop:Ysemigroup}, it does not make direct use of the Poissonian nature of the jump mechanism.

\section{Malthusian properties of the associated semi-group of operators}
\label{sec:Malthus}

In this section, we check that our model defines a \textit{Malthusian process}, as coined by Olofsson~\cite{Olofsson2009}. 
For the sake of simplicity, we assume in this section that $\Sbb$ is compact, and that there exists $b_m<b_M\in\R_+$ such that $b(s,a)$ is positive for all $(s,a)\in\Sbb\times [b_m,b_M]$. We also assume that, for all $s_0,a_0\in\Sbb\times\R_+$,  $\gamma_{s_0,a_0}$ admits a positive density with respect to a reference  measure $\pi$ on $\Sbb$, denoted by $g_{s_0,a_0}(s)\in(0,+\infty)$ at point $s$, and which is continuous with respect to $(s_0,s,a_0)\in\Sbb^2\times\R_+$.

The Malthusian parameter associated  to the process $X$ (see~\cite[Section~5]{Jagers1989}) is given by
\[
	\alpha:=\inf\big\{\lambda\in\R,\ \text{such that}\ \nu_{\lambda}(s,S\times\R_+)<+\infty\ \text{for some}\ s\in\Sbb\big\}.
\]
Defining the semi-group $(R_t)_{t\geq 0}$  by
\[
\delta_{(s_0,a_0)} R_t:=\E_{(s_0,a_0)}\left(\ind_{Z_t\in\cdot}\exp\left(\int_0^t b(Z_u)\,\mathrm{d}u\right)\right),\quad \forall (s_0,a_0)\in\Sbb\times \R_+,
\]
we prove that $\alpha$ is equal to the leading eigenvalue $-\lambda_0$ of the semi-group $(R_t)_{t\geq 0}$ when $\lambda_0$ is negative. 
 We first state some spectral properties of $R$, including the existence of $\lambda_0$, related to the theory of quasi-stationary distributions (we refer the reader to~\cite{ColletMartinezEtAl2013,DoornPollett2013,MeleardVillemonais2012} for general references to quasi-stationary distributions). The proof of the following result is postponed to Section~\ref{sec:proofM1}.

\begin{theorem}
	\label{thm:malthus1}
	Under the above assumptions, there exists a non-negative measurable function $\eta:\Sbb\times\R_+\to\R_+$,  constants $\lambda_0\in\R$, $\lambda_1\in(0,+\infty)$, $C>0$ and a probability measure $\Upsilon$ on $\Sbb\times \R_+$ such that, for all $t\geq 0$ and all $(s_0,a_0)\in\Sbb\times\R_+$,
	\[
		\left\|e^{\lambda_0 t}\delta_{(s_0,a_0)}R_t-\eta(s_0,a_0)\Upsilon\right\|_{TV}\leq Ce^{-\lambda_1 t}.
	\]
	Moreover, if $\lambda_0<0$, the Malthusian parameter $\alpha$ of the branching process $X$ equals $-\lambda_0$.
\end{theorem}

\begin{remark}
    The proof of Theorem~\ref{thm:malthus1} can be adapted to a more general setting (for instance with $\Sbb$ not compact, or $b(s,\cdot)$ positive on a segment that depends on $s$), at the expense of additional technicalities both in the presentation of the assumptions and in the proofs. 
\end{remark}


For any measurable function $f:\Sbb\times\R_+\to\R$ and all $t\geq 0$, we define the function
\[
	f_t(s_0,a_0)=\ind_{a_0\leq t} f(s_0,t-a_0),
\]
so that, if $s_0$ and $a_0$ are respectively the type and the birth time of one individual, the number $f_t(s_0,a_0)$ is the function $f$ of the type and age of the individual at time $t$. This is similar to the \textit{$\chi$-counted population} introduced in Section~7 of~\cite{Jagers1989}. The next result thus describes the evolution of the expectation of the types and ages distribution across the population and is an immediate corollary of Theorem~\ref{thm:malthus1} and Proposition~\ref{prop:Zmanytoone} (recall that $X_n$ denotes the empirical measure of types and ages in the population at generation $n$). In particular, in this situation, it shows that  the population size evolves exponentially fast with exponential parameter $-\lambda_0$ and that the telomere length distribution across the population converges to a limit which does not depend on the initial distribution of the population ages and telomere lengths.

\begin{corollary}
	\label{cor:main}
	For all bounded measurable function $f:\Sbb\times\R_+\to\R$ and all $t\geq 0$,
	\begin{align*}
		\E_{(s_0,a_0)}\left(\sum_{n\geq 0} X_n(f_t)\right)=\delta_{(s_0,a_0)}R_t f.
	\end{align*}
	In particular, for all $s_0\in\Sbb$, 
	\begin{align*}
		\left|e^{\lambda_0 t} \E_{(s_0,a_0)}\left(\sum_{n\geq 0} X_n(f_t)\right)-\eta(s_0,a_0)\Upsilon(f)\right|\leq C e^{-\lambda_1' t}\|f\|_\infty,
	\end{align*}
	where $\lambda_0,\Upsilon$ are from Theorem~\ref{thm:malthus1} and for some constants $C,\lambda_1'>0$. 	
\end{corollary}


%

We expect that  $e^{-\lambda_0 t} \sum_{n\geq 0} X_n(f_t)$ converges almost surely toward $\eta(s_0,a_0)\Upsilon(f)$ times a non-negative random variable. This type of results is classical in the setting of multi-types branching processes and for branching processes with irreducible reproduction measure (see for instance~\cite{Olofsson2009}). The extension of these results to the situation at hand (where the reproduction measure is allowed to be reducible) is the subject of an ongoing work.

%
%
%

\section{Numerical simulations}

\label{sec:sim}

We analyze numerically the influence of  the attrition parameter $\mu$ and  the birth rate curve on the limit distribution $\Upsilon$, and on the dynamic of the telomere length distribution across the population.
We set: $\Sbb=[l_{\min},l_{\max}]$, where $l_{\min}=5$ kbp and $l_{\max}=25$ kbp; $\alpha=0.017$; 
\[
\gamma_{(s_0,a)}(\dd s)=c_\gamma\,\ind_{s\in[l_{\min},l_{\max}]}\exp\left({-\frac{(s-(s_0+(a-a_p)\alpha-\mu))^2}{2\sigma^2}}\right)\,\dd s,\ \forall (s_0,a)\in\Sbb\times\R_+\to\R_+,
\]
where $\sigma=0.1$, $c_\gamma=1/\int_{l_{\min}}^{l_{\max}} e^{-\frac{(u-\mu)^2}{2\sigma^2}}\,\dd u$ and $\mu=20\alpha$ (different values of $\mu$ in  $\{0,\ldots,25 \alpha\}$ will also be considered); and
\[
b(a)=b_0(a+s_f),\ \forall a\geq 0,
\]
where $b_0$ is chosen according to the demographical empirical curve of year 1960 (see Figure~\ref{fig:birthrate}) and $s_f=0$ is a shift parameter (different values of $s_f$ in $[-10,10]$ will be considered).


\medskip 

We first investigate the influence  of the parameters $\mu$ (see Figure~\ref{fig:Upsilon_fct_mu}) and $s_f$ (see Figure~\ref{fig:Upsilon_fct_fs}) on the long time limiting distribution $\Upsilon$ (see Theorem~\ref{thm:malthus1}), which should be interpreted as the equilibrium telomere length distribution. As expected, a higher attrition before reproduction age (i.e. a higher parameter $\mu$) leads to lower telomere lengths in the population at equilibrium, while a higher parental age at birth (i.e. a higher parameter $s_f$) entails higher telomere lengths in the population at equilibrium. An important feature of the model is that the equilibrium distribution is highly concentrated on the boundaries of the admissible limits. 
Figure~\ref{fig:moyenne_finale_fonction_mu} (resp. Figure~\ref{fig:moyenne_finale_fonction_fs})  displays the influence of $\mu$ (resp.~$s_f$) on the population's mean telomere length at equilibrium. We observe that the influences of $\mu$ and $s_f$ are non-linear. Depending on the parameters value, even a slight  decrease in the attrition parameter $\mu$ or an increase in the parental age at birth $s_f$ can have drastically different effects; the model displays a transition phase phenomenon, with approximate critical values $\mu=12.5$ and $s_f=15$.

\begin{figure}
    \centering
    \includegraphics[width=10cm]{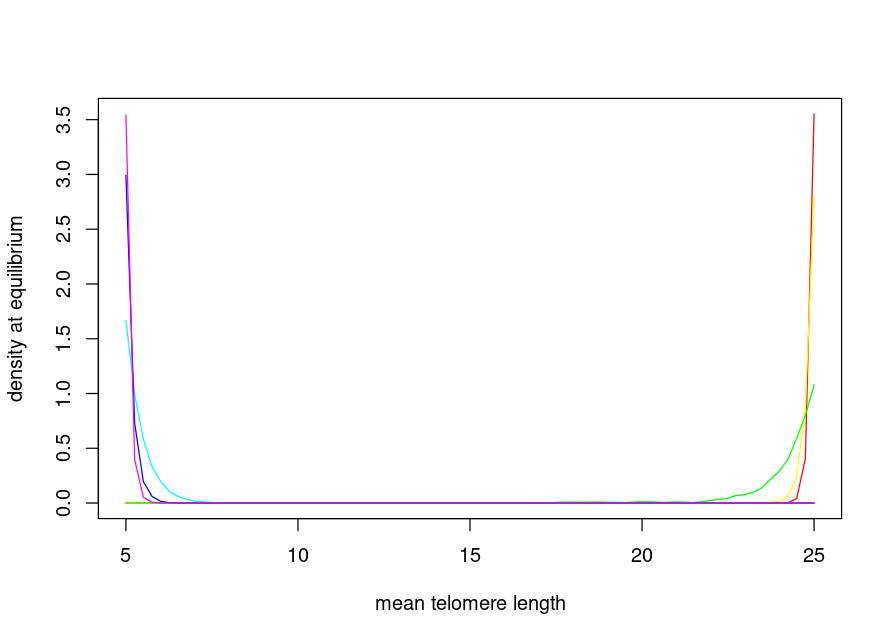}
    \caption{\label{fig:Upsilon_fct_mu} Density of the telomere's length distribution for values of $\mu$ in $\{8\alpha,10\alpha,12\alpha,14\alpha,16\alpha,18\alpha\}$, where $\alpha=0.017$ is fixed ($\mu=8\alpha$ corresponds to the rightmost curve and $\mu=18\alpha$ to the leftmost curve).}
\end{figure}

\begin{figure}
    \centering
    \includegraphics[width=10cm]{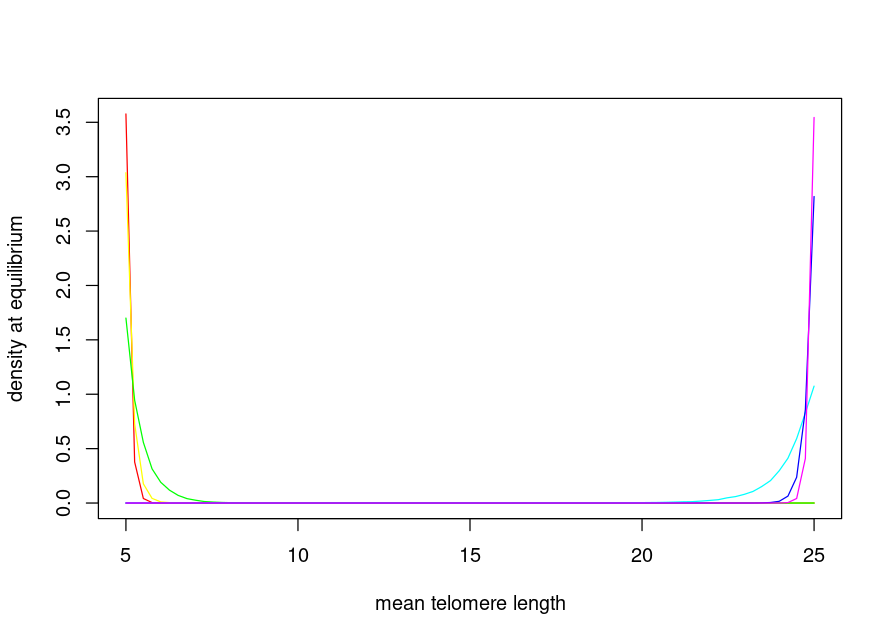}
    \caption{\label{fig:Upsilon_fct_fs}  Density of the telomere's length distribution for values of $s_f$ in $\{10,12,14,16,18,20\}$ ($s_f=10$  corresponds to the leftmost curve and $s_f=20$ to the rightmost curve).}
\end{figure}

\begin{figure}
    \centering
    \includegraphics[width=10cm]{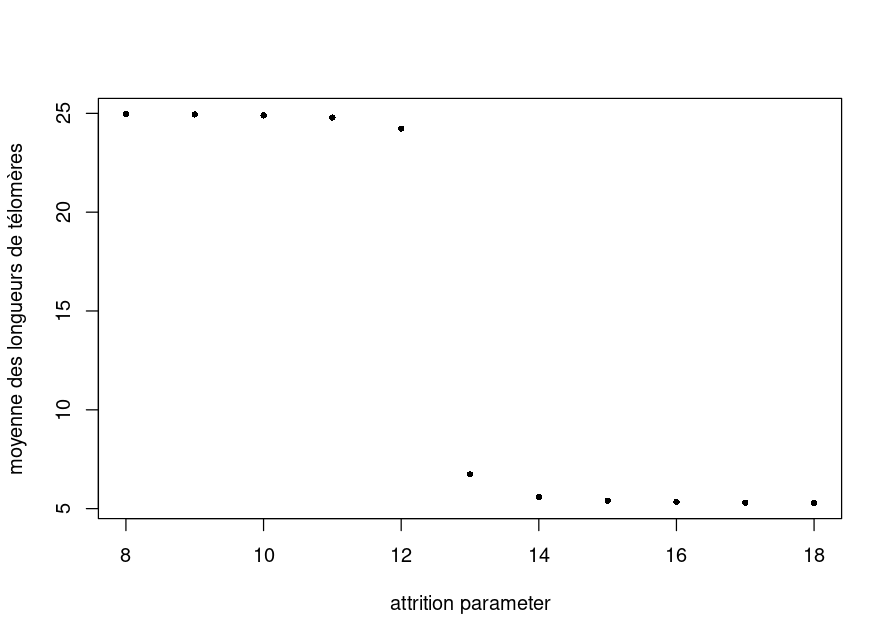}
    \caption{\label{fig:moyenne_finale_fonction_mu} Population mean telomere length at equilibrium, as a function of $\mu/\alpha$, when $\mu$ ranges from $8\times \alpha$ to $18\times \alpha$, and $\alpha = 0.017$ is fixed.}
\end{figure}

\begin{figure}
    \centering
    \includegraphics[width=10cm]{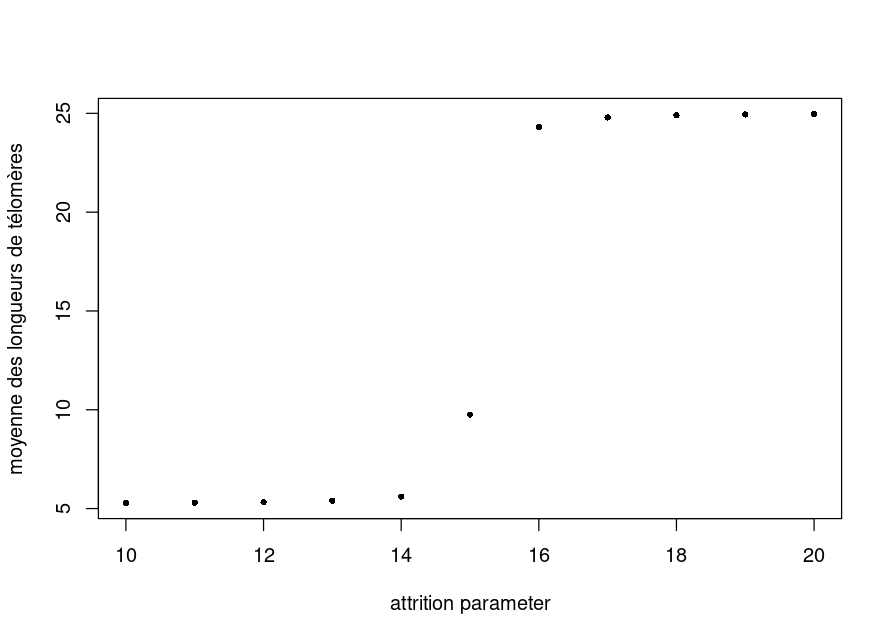}
    \caption{\label{fig:moyenne_finale_fonction_fs} Population mean telomere length at equilibrium, as a function of $s_f$, when $s_f$ ranges from $10$ to $20$.}
\end{figure}

\medskip
To investigate the time evolution of the telomere length in the population, we study the speed of convergence toward the equilibrium distribution $\Upsilon$ (see Theorem~\ref{thm:malthus1}) and the time evolution of the mean telomere length in the population. We display the total variation distance between the population telomere length distribution at time $t$ and the telomere length distribution at equilibrium, given an initial population of individuals with age $0$ and telomere length $18$ kbp. Figure~\ref{fig:Distance_Upsilon_fct_mu} (resp. Figure~\ref{fig:Distance_Upsilon_fct_fs}) displays the evolution of this distance to equilibrium for different values of $\mu$ (resp.~$s_f$). We observe that the speed of convergence to the equilibrium depends on the parameters, and that, in all cases, the  convergence to equilibrium arises after several thousands years. Figure~\ref{fig:Mean_evolution_fct_mu} (resp. Figure~\ref{fig:Mean_evolution_fct_fs}) displays the evolution of the population's mean telomere length over time for several values of $\mu$ (resp.~$s_f$). The mean  also stabilises after several thousand years for most parameter choices, and a drift in the telomere length can be sustained for several thousand years. In such a time frame,  mean attrition before reproduction and parental age at birth is subject to important changes, because of the demographic evolution.  

As a result, our model suggests that the limiting distribution $\Upsilon$ does not materialize in a population where the parameters may change in a time window of less than one thousand year, so that, in empirical measures, the population is not observed at equilibrium, and that 
the drift in the population's mean telomere length can be sustained during very long periods of time. This is coherent with the findings of~\cite{HolohanDeMeyerEtAl2015}.

\begin{figure}[h]
    \centering
    \includegraphics[width=10cm]{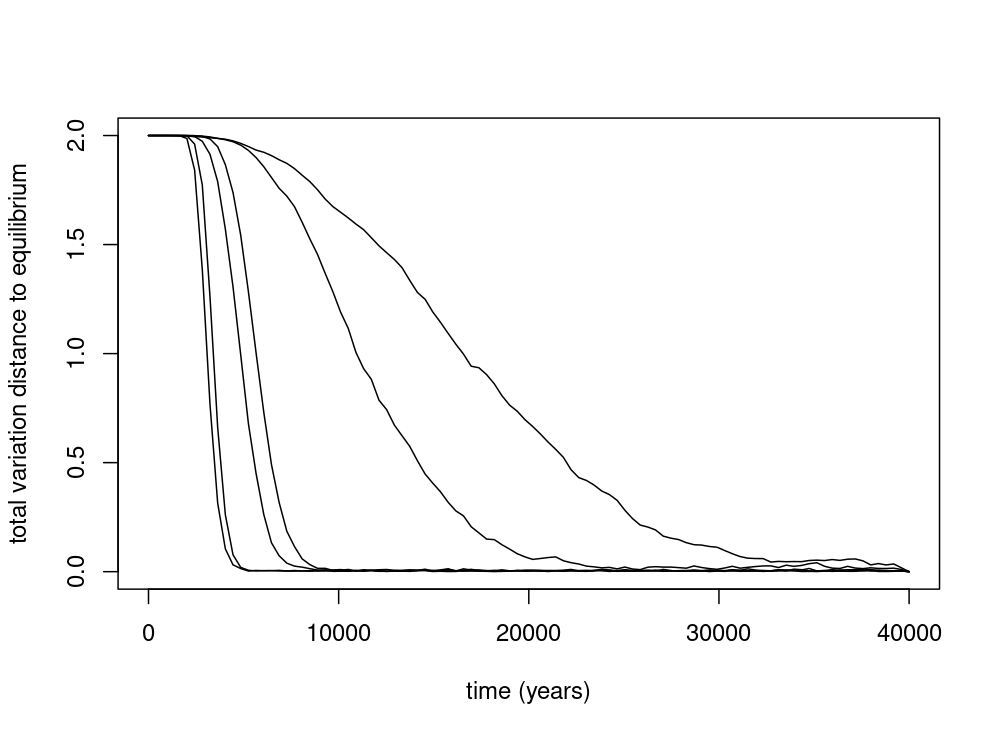}
    \caption{\label{fig:Distance_Upsilon_fct_mu} Total variation distance as a function of time between the population's telomere distribution and telomere length distribution at equilibrium. Values of $\mu$ are in $\{8\alpha,10\alpha,12\alpha,14\alpha,16\alpha,18\alpha\}$, with  $\alpha=0.017$.}
\end{figure}

\begin{figure}[h]
    \centering
    \includegraphics[width=10cm]{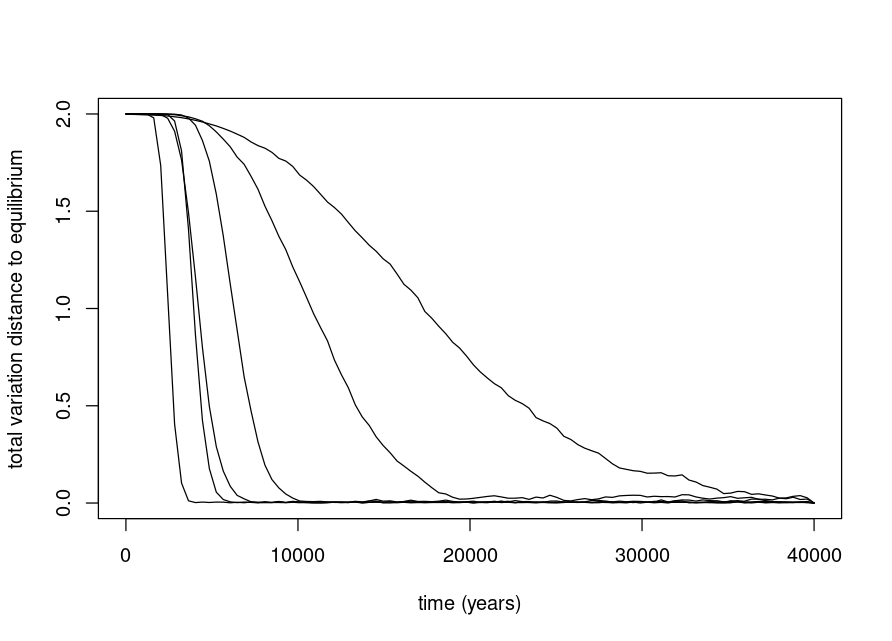}
    \caption{\label{fig:Distance_Upsilon_fct_fs} Total variation distance as function of time between the population's telomere length distribution and the telomere length distribution at equilibrium. Values of $s_f$ are in $\{10,12,14,16,18,20\}$.}
\end{figure}

\begin{figure}[h]
    \centering
    \includegraphics[width=10cm]{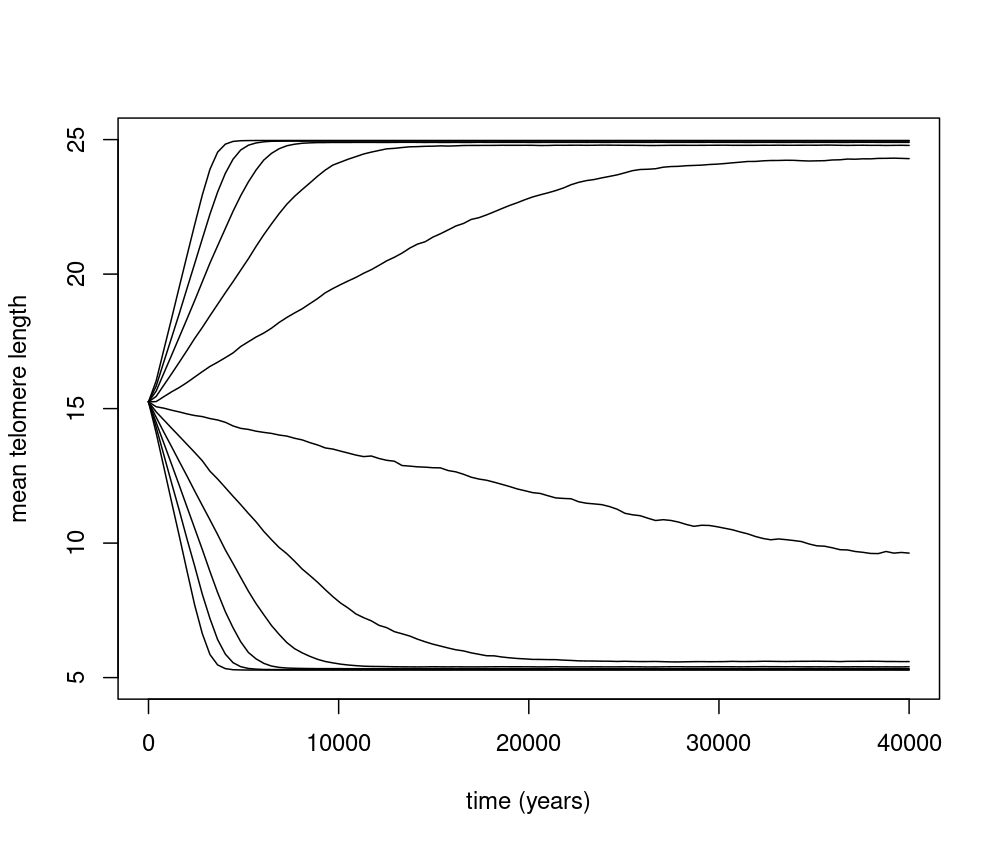}
    \caption{\label{fig:Mean_evolution_fct_mu} Population mean telomere length as a function of time, for different values of $\mu$. The highest trajectory corresponds to $\mu=8\alpha$ and the lowest trajectory to $\mu=18\alpha$, with $\alpha=0.017$.}
\end{figure}

\begin{figure}[h]
    \centering
    \includegraphics[width=10cm]{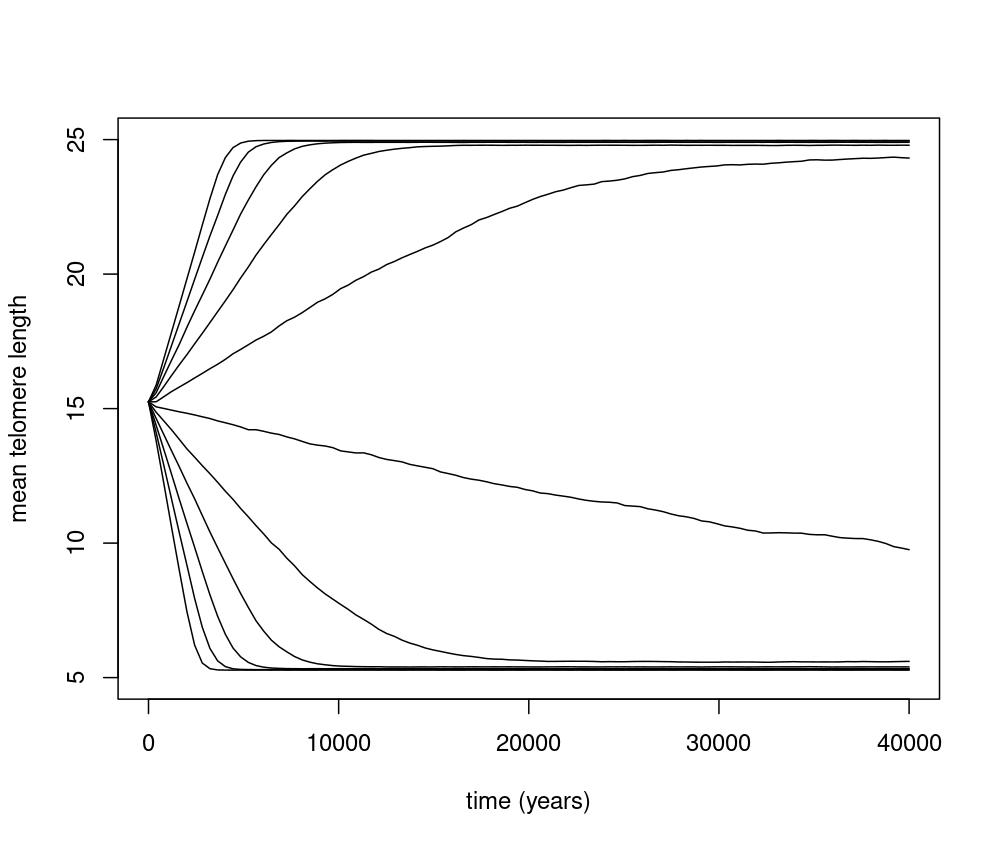}
    \caption{\label{fig:Mean_evolution_fct_fs} Population mean telomere length as a function of time, for different values of $s_f$. The highest trajectory corresponds to $s_f=20$ and the lowest trajectory to $s_f=10$.}
\end{figure}

Finally, we investigate the situation where the parameters $\mu$ or $s_f$ undergo a change at some time point. Namely, we consider the situation where $\mu$ or $s_f$ are constant for $5000$ years and then shift to a new value for $1000$ years. Our main focus is to study the impact of this evolution on the relations among individuals between telomere length at puberty ($s_0$), father age at birth (FAB), parental birth year (PBY) and birth date (BD), and to compare the behaviour of our model to the findings of~\cite{HolohanDeMeyerEtAl2015}.  Figure~\ref{fig:mu_evolves_PBY_adj_FAB} displays the evolution of the individual telomere length as a function of PBY adjusted for FAB, when $\mu=13$ for individuals born in the time interval $[0,5000]$ and $\mu=16$ for individuals born in the time interval $(5000,6000]$. Figure~\ref{fig:mu_evolves_BD_adj_FAB_PBY} displays the evolution of  the individual telomere length as function of BD adjusted for PBY and FAB for individuals born after time $5000$. Contrarily to the findings of~\cite{HolohanDeMeyerEtAl2015}, we do not observe a significant negative correlation in this last figure (different time frames do not change this fact). The negative impact of the change of $\mu$ at time $5000$ does not materialize in this plot and thus our model can not support the hypothesis of~\cite{HolohanDeMeyerEtAl2015} that an event negatively impacted the human population telomere length a century ago. This suggests that other mechanisms (either a statistical artefact or another type of demographical event) are the cause of the negative correlation found in the above cited study.

\begin{figure}[h]
    \centering                                                                                               \includegraphics[width=10cm]{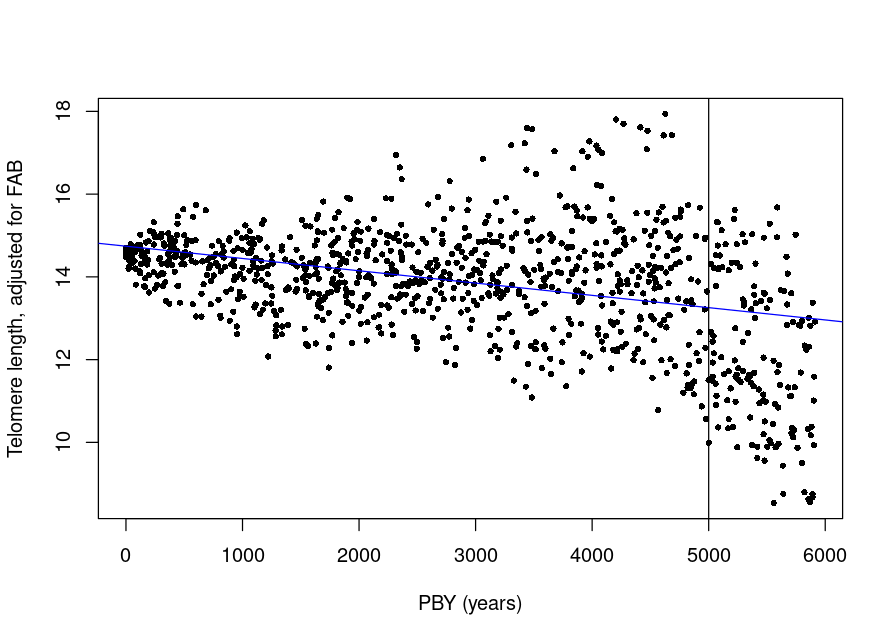}
    \caption{\label{fig:mu_evolves_PBY_adj_FAB} Sample of the population and linear regression lines for $s_0$ as a function of PBY adjusted for FAB, when $\mu=13$ for individuals born in the time interval $[0,5000]$ and $\mu=16$ for individuals born in the time interval $(5000,6000]$. The regression line has coefficient $-2.97\times 10^{-4}$.}
\end{figure}

\begin{figure}[h]
    \centering                                                                                               \includegraphics[width=10cm]{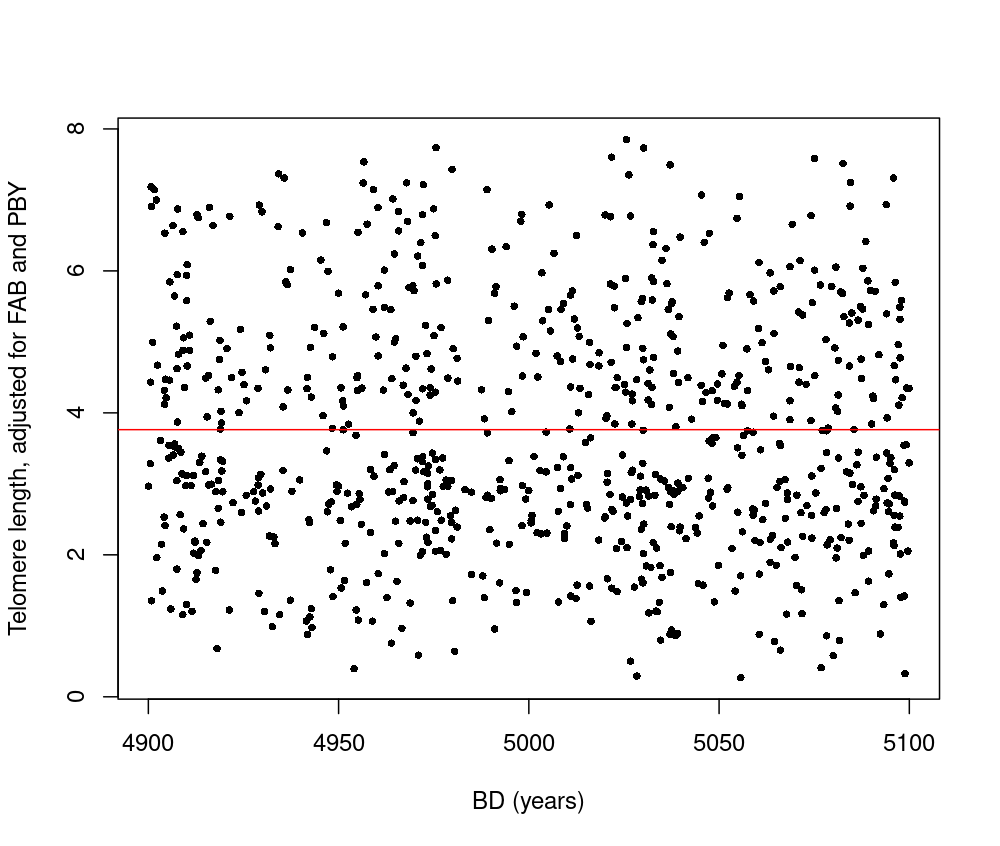}
    \caption{\label{fig:mu_evolves_BD_adj_FAB_PBY} Sample of the population and linear regression lines for $s_0$ as a function of BD adjusted for FAB and PBY, when $\mu=13$ for individuals born in the time interval $[0,5000]$ and $\mu=16$ for individuals born in the time interval $(5000,6000]$. The regression line has coefficient $-2.97\times 10^{-4}$.}
\end{figure}

\section{Conclusion}
We constructed a probabilistic model representing the evolution of telomere length in a population across multiple generations. Various mathematical results, including many-to-one formulas and Perron-Frobenius type results, have allowed us to exhibit interesting properties concerning the asymptotic behaviour of the average telomere length in a population. These results were confirmed empirically by experiments in silico. In particular, we found the definite influence of the attrition parameter, as well as that of a particular modification of the reproduction rate using a time shift of the fertility curve.

We also studied the link between the length of the telomeres of individuals and the date of birth of their ancestor, adjusted by its age at the birth of the descendants. We were able to compare these numerical results with the literature. In particular, we could not confirm that an increase in the attrition parameter led to a negative correlation between telomere length and date of birth adjusted by FAB and PBY, as observed in~\cite{HolohanDeMeyerEtAl2015}. Understanding the discrepancy between the empirical and numerical datas will require further investigations.

The proposed model contains important simplifications; it therefore appears necessary to study richer models. The integration of heterogeneity within the population with attrition factors depending on the geographical or societal environment would be more realistic, using a richer type space $\Sbb$ (see also \cite{HaccouHaccouEtAl2005} and references therein) or building a model using branching processes in random environment (see e.g.~\cite{Bansaye2009,Kaplan1974,Keiding1976,KerstingVatutin2017,SmithWilkinson1969}). Taking into account the influence of certain migratory phenomena would also be interesting (see for instance~\cite{BansayeMeleard2015,KawazuWatanabe1971a,Li2006,Pakes1971}). Finally, it seems essential to move towards the implementation of a bisexual model (see e.g.~\cite{Daley1968,DaleyHullEtAl1986,FritschVillemonaisEtAl2022,GonzalezMolina1996,Hull2003,Molina2010}).

\section*{Acknowledgement}
This work was supported partly by the french PIA project « Lorraine Université d’Excellence », reference ANR-15-IDEX-04-LUE. 

\section*{Competing interest}
The authors have no relevant financial or non-financial interests to disclose.

\section*{Data availability statement}
No datasets were generated or analysed during the current study.

\newpage
\appendix
\section{Proof of the results of Sections~\ref{sec:poissonian} and~\ref{sec:nonpoissonian}}

We first prove Propositions~4 and~5 from Section~\ref{sec:nonpoissonian} and conclude with Proposition~3 from Section~\ref{sec:poissonian}.

\subsection{Proof of Proposition~\ref{prop:Ysemigroup}}
\label{sec:proof1}

We have, setting $f=\ind_{A\times B}$,
\begin{align*}
\E_{0, y_0}\left(f(N_{t+u},Y_{t+u})\ |\ \mathcal{G}_t\right) &= \sum_{k=0}^{+\infty} \sum_{i = k}^{+ \infty} \E_{0, y_0}\left(\ind_{N_t = k}\ \ind_{N_{t+u} = i}\ f(i,Y_{t+u})\ |\ \mathcal{G}_t\right).
\end{align*} 

Fix $i,k\in\Z_+$. Assume first that $i\geq k+1$, and set $j = i - (k+1)$. We define the $\sigma$-algebra $\cF_{k+1}=\sigma\big((\sigma_0,W_0),(\sigma_1,W_1),\ldots,(\sigma_{k+1},W_{k+1})\big)$.
Then, for any $\mathcal{G}_t$-measurable non-negative random variable $Z$, the random variable $Z\ind_{N_t=k}$ is $\cF_{k+1}$-measurable and hence
\begin{align*}
\E_{0,y_0}&\left(Z\,\ind_{N_t = k}\,\ind_{N_{t+u} = i}\ f(i,Y_{t+u})\right)\\ 
&= \E_{0,y_0}\left[\E\left(Z\,\ind_{N_t = k}\,\ind_{N_{t+u} = i}\ f(i, Y_{t+u})\,|\,\cF_{k+1}\right)\right]\\
&= \E_{0,y_0}\left[Z\,\ind_{N_t = k}\ \E\left(\ind_{\sigma_i\leq t+u<\sigma_{i+1}}\ f(i,\vartheta_{t+u-\sigma_i}W_i) \,|\,\cF_{k+1}\right)\right]\\
&= \E_{0,y_0}\left[Z\,\ind_{N_t = k}\ \E^W_{\sigma_{k+1},W_{k+1}}\left(\ind_{\sigma_j\leq t+u<\sigma_{j+1}}\ f(i,\vartheta_{t+u-\sigma_j}W_j)\right)\right],
\end{align*}
where $\E^W_{u,w}$ is the expectation with respect to the law of $(\sigma_n,W_n)_{n\geq 0}$ when $\sigma_0=u$ and $W_0=w$. Now, since $\sigma_{k+1}=t+\tau(Y_t)$, we have
\begin{align*}
\E_{0,y_0}&\left(Z\,\ind_{N_t = k}\,\ind_{N_{t+u} = i}\ f(i,Y_{t+u})\right)\\ 
&= \E_{0,y_0}\left[Z\,\ind_{N_t = k}\ \E^W_{t+\tau(Y_t),W_{k+1}}\left(\ind_{\sigma_j\leq t+u<\sigma_{j+1}}\ f(i,\vartheta_{t+u-\sigma_j}W_j)\right)\right]\\
&= \E_{0,y_0}\left[Z\,\ind_{N_t = k}\ \E^W_{\tau(Y_t),W_{k+1}}\left(\ind_{\sigma_j\leq t+u<\sigma_{j+1}}\ f(i,\vartheta_{u-\sigma_j}W_j)\right)\right].
\end{align*}
Using the fact that $W_{k+1}\sim \Pi(Y_t,\cdot)$, we deduce that
\begin{align*}
\ind_{N_t = k}\,\E&\left(\ind_{N_{t+u} = i}\ f(i,Y_{t+u})\,\mid \cG_t\right)= \ind_{N_t = k}\ \int_{\X}\,\Pi(Y_t,\dd w)\, \E^W_{\tau(Y_t),w}\left(\ind_{\sigma_j\leq t+u<\sigma_{j+1}}\ f(i,\vartheta_{u-\sigma_j}W_j)\right).
\end{align*}
Assume now that $i=k\in\Z_+$, so that
\[
 \E\left[\ind_{N_t = k}\ \ind_{\sigma_i \leq t + u < \sigma_{i+1}}\ f(i, Y_{t+u})\ |\ \mathcal{G}_t\right] = \ind_{N_t = k}\ \ind_{u < \tau(Y_t)}\ f(k,\vartheta_u Y_t)
 \]
Using the last two equations, we deduce that, for all $k\geq 0$,
\begin{align*}
 \E\left(\ind_{N_t = k}\ f(i, Y_{t+u})\ |\ \mathcal{G}_t\right) &= \ind_{N_t = k}\ \Big[ \ind_{u < \tau(Y_t)}\ f(k, \vartheta_u Y_t) \\
 &\quad\quad+ \sum_{j=0}^{+\infty} \int_{\X}\,\Pi(Y_t,\dd w)\, \E^W_{\tau(Y_t),w}\left(\ind_{\sigma_j\leq t+u<\sigma_{j+1}}\ f(N_t+1+j,\vartheta_{u-\sigma_j}W_j)\right) \Big]
 \end{align*}
Summing over $k \geq 0$, we finally obtain
\begin{multline*}
 \E\left(f(Y_{t+u})\ |\ \mathcal{G}_t\right) = \ind_{u < \tau(Y_t)}\ f(N_t, \vartheta_u Y_t)\\ +  \sum_{j=0}^{+\infty} \int_{\X}\,\Pi(Y_t,\dd w)\, \E^W_{\tau(Y_t),w}\left(\ind_{\sigma_j\leq t+u<\sigma_{j+1}}\ f(N_t+1+j,\vartheta_{u-\sigma_j}W_j)\right),
\end{multline*}
which concludes the proof.

\subsection{Proof of Proposition~\ref{prop:Ymanytoone}}
\label{sec:proof2}

Let us introduce the sequence of random indices $(I_n)_{n\geq 1}$, defined inductively by $I_0=0$ and, for all $n\geq 0$,
\[
	I_{n+1}=\inf\{k\geq I_n+1,\text{ such that }\varepsilon_k=1\}.
\]
One easily checks that, with this notation, $W_{I_n}^{(a)}=0$ for all $n\geq 1$ almost surely and that, for all $t\geq 0$,
\begin{equation}
\label{eq:YnewDecomp}
	(Y^{(s)}_t,Y^{(a)}_t)=\sum_{n\geq 0} \ind_{\tau_n\leq t<\tau_{n+1}}(V^{(s)}_n,t-\tau_n),
\end{equation}
where $\tau_n:=\sigma_{I_n}$ and $V_s=W_{I_n}$ (so that $V^{(s)}_n:=W^{(s)}_{I_n}$). We prove the result by induction on $n$.

\medskip\noindent
\textit{Step~1. Initialization of the inductive procedure.}\\
For all $w_0=(s_0,0,\xi_0)\in \X$, we have, setting $f=\ind_{A\times B}$,
\begin{align*}
\E^{\sigma,W}_{0,w_0}\left[f(V^{(s)}_0,\tau_0)\,2^{N_t}\,\ind_{\tau_0\leq t<\tau_1}\right]
&=f(w_0,0)\,\E^{\sigma,W}_{0,w_0}\left[2^{N_t}\,\ind_{\tau_0\leq t<\tau_1}\right]=f(w_0,0).
\end{align*}
Then, for all $w_0=(s_0,0,\xi_0)\in \X$, we have, setting $f=\ind_{A\times B}$,
\begin{align*}
\E^W_{0,w_0}\left[f(V^{(s)}_1,\tau_1)\,2^{N_t}\,\ind_{\tau_1\leq t<\tau_2}\right]
&= \E^W_{0,w_0}\left[f(V^{(s)}_1,\tau_1)\,2^{N_{\tau_1}}\,\ind_{\tau_1\leq t}\,\E^W_{\tau_1,V_1}\left(2^{N_{t-u}}\ind_{t-u<\tau_1}\right)_{\rvert u=\tau_1}\right]\\
&=\E^W_{0,w_0}\left[f(V^{(s)}_1,\tau_1)\,2^{N_{\tau_1}}\,\ind_{\tau_1\leq t}\right]\\
&=\sum_{k\geq 0}  \E^W_{0,w_0}\left[\ind_{I_1=k}f(\xi^{(s)}_0(k),\xi^{(a)}_0(k))\,2^k\,\ind_{\xi^{(a)}_0(k)\leq t}\right]\\
&= \sum_{k\geq 0} f(\xi^{(s)}_0(k),\xi^{(a)}_0(k))\,\ind_{\xi^{(a)}_0(k)\leq t}\,2^k\,\P^W_{0,w_0}(I_1=k)\\
&= \sum_{k\geq 0} f(\xi^{(s)}_0(k),\xi^{(a)}_0(k))\,\ind_{\xi^{(a)}_0(k)\leq t}.
\end{align*}
Integrating with respect to $\P^J_{s_0}$ and using~\eqref{eq:YnewDecomp}, we deduce that
\[
	\int_{\M}\P^J_{s_0}(d\xi_0)\E^Y_{(s_0,0,\xi_0)}\left[f(Y^{(s)}_t,t-Y^{(s)}_t)\,2^{N_t}\,\ind_{\tau_1\leq t<\tau_2}\right]=\mu(s_0,A\times B\cap[0,t]).
\]
This concludes the first step, since $\sum_{k=0}^{N_t} \varepsilon_k=1$ is equivalent to $\tau_1\leq t<\tau_2$.

\medskip\noindent
\textit{Step~2. Induction.}\\
Fix $n\geq 1$ and assume that the property holds true for this value of $n$. For all $w_0=(s_0,0,\xi_0)\in \X$, we have
\begin{multline*}
	\E^W_{0,w_0}\left[f(V^{(s)}_{n+1},\tau_{n+1})\ind_{\tau_{n+1}\leq t<\tau_{n+2}}2^{N_t}\right]\\
	\begin{aligned}
	&=\E^W_{0,w_0}\left[\ind_{\tau_1\leq t} 2^{N_{\tau_1}}\,\E^W_{\tau_1,V_1}\left(f(V^{(s)}_{n},\tau_{n})\ind_{\tau_{n}\leq t<\tau_{n+1}}2^{N_{t-u}}\right)_{\rvert u=\tau_1}\right]\\
	&=\E^W_{0,w_0}\left[\ind_{\tau_1\leq t} 2^{N_{\tau_1}}\,\E^W_{0,V_1}\left(f(V^{(s)}_{n},\tau_{n}+u)\ind_{\tau_{n}\leq t-u<\tau_{n+1}}2^{N_{t-u}}\right)_{\rvert u=\tau_1}\right]\\
	&=\E^W_{0,w_0}\left[\ind_{\tau_1\leq t} 2^{N_{\tau_1}}\,\mu^n\big(V^{(s)}_1,A\times ((B-\tau_1)\cap[0,t-\tau_1])\big)\right],
	\end{aligned}
\end{multline*}
where we used the law of $V_1$ conditionally to $V_1^{(s)}$ and the induction assumption. Similarly as in Step~1, we now decompose over the possible values of $I_1$ and obtain
\begin{multline*}
\E^W_{0,w_0}\left[f(V^{(s)}_{n+1},\tau_{n+1})\ind_{\tau_{n+1}\leq t<\tau_{n+2}}2^{N_t}\right]\\
=\sum_{k\geq 0} \mu^n\big(\xi^{(s)}_0(k),A\times ((B-\xi^{(a)}_0(k))\cap[0,t-\xi^{(a)}_0(k)])\big)\,\ind_{\xi^{(a)}_0(k)\leq t}.
\end{multline*}
Using~\eqref{eq:YnewDecomp} and integrating with respect to $\P^J_{s_0}$, we obtain
\begin{align*}
\int_{\M}\P^J_{s_0}(d\xi_0)\E^Y_{(s_0,0,\xi_0)}&\left[f(Y^{(s)}_t,t-Y^{(s)}_t)\,2^{N_t}\,\ind_{\tau_{n+1}\leq t<\tau_{n+2}}\right]\\
&=\int_{\Sbb\times\R_+} \mu(s_0,\dd s\times \dd u) \, \mu^n\big(s,A\times ((B-u)\cap[0,t-u])\big)\,\ind_{u\leq t}\\
&=\int_{\Sbb\times\R_+} \mu(s_0,\dd s\times \dd u) \, \mu^n\big(s,A\times (B\cap[0,t]-u)\big)\\
&=\mu^{n+1}(s_0,A\times (B\cap[0,t])).
\end{align*}
This concludes the proof.

\subsection{Proof of Proposition~\ref{prop:Zmanytoone}}
\label{sec:proof3}

We denote by $(R_t)_{t\geq 0}$ the semi-group associated to the Feynman-Kac expression of the proposition. Namely, for all $f\in L^\infty(\Sbb\times\R_+)$ and all $(s_0,t_0)\in\Sbb\times\R_+$, we set
\[
R_t f(s_0,t_0)=\E^Z_{s_0,t_0}\left[f(Z_t)\ \exp\left(\int_0^t b\big(Z_u\big)\,\dd u\right)\right].
\]
In the setting of Poissonian reproduction mechanism of Section~\ref{sec:poissonian}, it is also clear that $\bZ:=(Y^{(s)},Y^{(a)})$ is a continuous time pure jump Markov process, with jump rate $b(\bZ)$ and jump distribution $\frac{1}{2}\left(\delta_{\bZ_t}+\gamma_{\bZ}\otimes\delta_0\right)$. We denote by $T$ the semi-group of defined by
\[
T_t f(s_0,t_0)=\E^{\bZ}_{s_0,t_0}\left(f(\bZ_t)\,2^{N_t}\right),
\]
were $N_t$ is the number of jumps of $\bZ$ before time $t$. Of course, this coincides with the last expression of Proposition~\ref{prop:Ymanytoone}.

Our aim is to show that $R$ and $T$ coincide when they are considered as semi-group over the set of uniformly continuous functions $\UC$. This will imply that they coincide on bounded measurable functions and hence, using Proposition~\ref{prop:Ymanytoone}, that the equality of Proposition~\ref{prop:Zmanytoone} holds true. 

We first show that both semi-groups can be restricted to $\UC$. 
\begin{lemma}
	For all $f\in \UC$ and all $t\geq 0$, $R_t f\in\UC$ and $T_t f\in\UC$.
\end{lemma}
\begin{proof}
	The process $(Z_t,\int_0^t b(Z_u)\,\dd u)_{t\geq 0}$ is a PDMP which satisfies the conditions of~\cite[Theorem~27.6]{Davis1993}. Hence, for any value of $m>0$, the application
	\[
		g_m:(s_0,t_0)\mapsto \E^Z_{s_0,t_0}\left[f(Z_t)\ \exp\left(m\wedge\int_0^t b\big(Z_u\big)\,\dd u\right)\right]
	\]
	is bounded and continuous. But $\int_0^t b\big(Z_u\big)\,\dd u$ is uniformly bounded by $t\,\|b\|_\infty$, which implies that, choosing $m$ large enough, $R_t f$ is continuous and bounded. Moreover, $b$ has compact support, say $K$. Hence, for all $(s_0,t_0)\notin K$, $(Z_t)_{t\geq 0}$ starting from $(s_0,t_0)$ is equal to $(s_0,t_0+t)$ almost surely and hence $R_t f(s_0,t_0)=f(s_0,t_0+t)$ which is uniformly compact over $(s_0,t_0)\notin K$. Since $R_t f$ is continuous over $K$, it is also uniformly continuous over $K$, and we deduce that $R_t f\in\UC$.
	
	Similarly, the process $(\bZ_t,N_t)_{t\geq 0}$ is a PDMP which satisfies the conditions of~\cite[Theorem~27.6]{Davis1993}. Hence, for any value of $m>0$, the application
	\[
	h_m:(s_0,t_0)\mapsto \E^\bZ_{s_0,t_0}\left[f(\bZ_t)\ 2^{m\wedge N_t}\right]
	\]
	is bounded an continuous. But 
	\[
	\left|\E^\bZ_{s_0,t_0}\left[f(\bZ_t)\ 2^{N_t}\right]-\E^\bZ_{s_0,t_0}\left[f(\bZ_t)\ 2^{m\wedge N_t}\right]\right|\leq \|f\|_\infty \E^\bZ_{(s_0,t_0)}\left| 2^{N_t}\ind_{N_t>m}\right|,
	\]
	where the right hand side converges to $0$ uniformly in $(s_0,t_0)\in \Sbb\times \R_+$ since $b$ is bounded. This implies that $T_t f$ is the uniform limit of a sequence of bounded continuous functions, and hence that it is itself bounded continuous. As in the case of $R_t f$, we deduce that $T_t f\in\UC$.
\end{proof}

We then obtain the strong continuity (see \cite[Definition~2.1]{Pazy1983}) of the semi-groups when they are defined over $\UC$.
\begin{lemma}
	\label{lem:strongcontinuity}
	The semi-groups $R$ and $T$ are strongly continuous semi-groups on $\UC$, which means that
	\[
		\left\|R_t f-f\right\|_\infty\xrightarrow[t\to0]{} 0,
	\]
	and similarly for $T$.
\end{lemma}

\begin{proof}
	Since the proof is similar for both semi-group, we only detail it for $T$. Let $f \in \UC$, $\varepsilon > 0$ and  $\delta > 0$ such that
	\[
	 \forall\  z_1,z_2,\  d(z_1,z_2) < \delta \Rightarrow |f(z_1)-f(z_2| < \varepsilon.
	 \]
	
	For any $z_0 = (s_0, a_0) \in \Sbb\times \R$ and $t < \delta$, we have
	\begin{align*}
	\left|T_t f(s_0, a_0) - f(s_0, a_0)\right| &= \left|\E_{s_0, a_0}\left(f(Z_t) 2^{N_t}\ind_{N_t = 0}\right) + \E_{s_0, a_0}\left(f(Z_t) 2^{N_t}\ind_{N_t \geq 1}\right) - f(s_0, a_0)\right|\\
	& \leq \left|f(s_0, a_0+t) \P(N_t = 0) - f(s_0, a_0)\right| + \|f\|_\infty\E_{s_0, a_0}\left(2^{N_t}\ind_{N_t \geq 1}\right) \\
	& \leq \left|f(s_0, a_0+t) - f(s_0, a_0)\right| + \|f\|_\infty\P(N_t \geq 1) + \|f\|_\infty\E_{s_0, a_0}\left(2^{N_t}\ind_{N_t \geq 1}\right).
	\end{align*}
	
	The term $\left|f(s_0, a_0+t) - f(s_0, a_0)\right|$ goes to $0$ when $t\to0$, uniformly in $(s_0,a_0)$, since $f$ is uniformly continuous. Moreover, $\E_{s_0, a_0}\left(2^{N_t}\ind_{N_t \geq 1}\right)$ is bounded (since $N_t$ is stochastically dominated by a Poisson random variable) and goes to $0$ since $N_t$ is non-decreasing right continuous at time $0$, with $N_0=0$.

This concludes the proof of Lemma~\ref{lem:strongcontinuity}.
\end{proof}

We now use the characterization of strongly continuous semi-groups by their infinitesimal generators (see~\cite{Pazy1983} Definition~1.1 and Theorem~2.6).

\begin{lemma}
	Fix $f \in \UC$. Then there exists a constant $C > 0$ such that
	\[ \left\|\frac{R_t f -f}{t} - \frac{T_t f -f}{t}\right\|_\infty \leq C\, t.\]
	In particular, $R$ and $T$ have the same infinitesimal generator over $\UC$ and hence they coincide.
\end{lemma}

\begin{proof}
Let $z_0 = (s_0, a_0) \in \Sbb\times\R_+$. We have 
	\begin{align*}
	R_t f(z_0)&=\E^Z_{z_0}\left(f(Z_t)\ind_{\overline{N}_t=0}e^{\int_0^t b(Z_u)\,\dd u}\right)\\
	&\quad\quad+\E^Z_{z_0}\left(f(Z_t)2\ind_{\overline{N}_t=1}e^{\int_0^t b(Z_u)\,\dd u}\right)\\
	&\quad\quad+\E^Z_{z_0}\left(f(Z_t)2^{N_t}\ind_{\overline{N}_t\geq 2}e^{\int_0^t b(Z_u)\,\dd u}\right)\\
	&=f(s_0,a_0+t)e^{\int_0^t b(s_0,a_0+u)\,\dd u}\,e^{-\int_0^t b(s_0,a_0+u)\,\dd u}\\
	&\quad\quad+\int_0^t\dd u\,b(s_0,a_0+u)e^{-\int_0^u b(s_0,a_0+v)\,\dd v}\\
	&\phantom{\quad\quad+\int_0^t\dd u\,b(s_0,a_0}\times \int_{\Sbb}  \gamma_{s_0, a_0+u}( \dd s) f(s, t-u)\,e^{\int_0^u b(s_0, a_0 + v)\dd v + \int_0^{t-u}b(s,v)\dd v}\\
	&\quad\quad+{\cal O}(t^2),
	\end{align*}
	with ${\cal O}(t^2)$ uniform in the initial position $z_0\in\Sbb\times \R_+$,
	where the first term is obtained from the definition of the process $Z$ before its first jump-time, the second term from the jump rate of the process and its jump measure, and the third one from the fact that $b$ is uniformly bounded and $N_t$ is stochastically dominated by a Poisson random variable with intensity $\|b\|_\infty$ (hence uniformly in $z_0\in\Sbb\times\R_+$).

Similarly,
	\begin{align*}
	T_t f(z_0) &= \E^\bZ_{z_0}\left(f(Z_t)\ 2^{N_t} \ind_{N_t = 0}\right) + \E^\bZ_{z_0}\left(f(Z_t)\ 2^{N_t} \ind_{N_t = 1}\right) + \E_{z_0}\left(f(Z_t)\ 2^{N_t} \ind_{N_t \geq 2}\right)\\
	&=f(s_0, a_0 + t)+\int_0^t \dd u \, b(s_0, a_0+u)\,e^{-\int_0^{u} b(s_0, a_0+v)\,\dd v} \int_{\Sbb} \gamma_{s_0, a_0 + u}( \dd s) f(s, t - u)+{\cal O}(t^2),
	\end{align*}
	where ${\cal O}(t^2)$ is uniform in $(s_0,a_0)$ (since the increment rate of $N_t$ is uniformly bounded, the probability that $2$ or more jump occur in time $t$ is of order $t^2$).

Hence we obtain
\begin{multline*}
\left|R_t f(z_0)-T_t f(z_0)\right|=\int_0^t\dd u\,b(s_0,a_0+u)e^{-\int_0^u b(s_0,a_0+v)\,\dd v} \int_{\Sbb}  \gamma_{s_0, a_0+u}(\dd s) |f(s, t-u)|\\
\times\left|e^{\int_0^u b(s_0, a_0 + v)\dd v + \int_0^{t-u}b(s,v)\dd v}-1\right|+{\cal O}(t^2).
\end{multline*}
But the term $\left|e^{\int_0^u b(s_0, a_0 + v)\dd v + \int_0^{t-u}b(s,v)\dd v}-1\right|$ is of smaller than $2\|b\|_\infty t$ for $t$ small enough (uniformly in $z_0\in\Sbb\times \R_+$), and we then deduce that 
\[
	\left\|R_t f(z_0)-T_t f(z_0)\right\|={\cal O}(t^2),
\]
which concludes the proof of the first assertion.

The fact that $R$ and $T$ have the same infinitesimal generator is a direct consequence of its definition (see ~\cite[Definition~1.1]{Pazy1983}). Since they are strongly continuous semi-group, we deduce from \cite[Theorem~2.6]{Pazy1983} that they are equal.

\end{proof}

We have proved that, for all fixed $t\geq 0$, $T_t f=R_t f$ for all functions $f\in\UC$. Since $\Sbb\times \R_+$ is a metric space, this implies that $T_t f=R_t f$ for all bounded Borel function $f$ (see for instance \cite[Lemma~2.3]{Varadarajan1958}). This and Proposition~\ref{prop:Ymanytoone} conclude the proof of Proposition~\ref{prop:Zmanytoone}.

\section{Proof of the results of Sections~\ref{sec:Malthus}}

This section is devoted to the proof of Theorem~\ref{thm:malthus1} and Corollary~\ref{cor:main}.

\subsection{Proof of Theorem~\ref{thm:malthus1}}
\label{sec:proofM1}
  
  Let $R^M$ be the sub-Markov semi-group defined as 
  \begin{align*}
  	\delta_{(s_0,a_0)} R_t^M :&=e^{-\|b\|_\infty t}\E^Z_{a_0,s_0}\left[\exp\left(\int_0^t
  	b(Z_u)\,du\right)\ind_{Z_t\in
  		\cdot}\ind_{t<\tau_M}\right]\\
  	&=\E^Z_{a_0,s_0}\left[\exp\left(\int_0^t
  	-\kappa(Z_u)\,du\right)\ind_{Z_t\in
  		\cdot}\ind_{t<\tau_M}\right],
  \end{align*}
  where $\tau_M=\inf\{t\geq 0,\ Z^{(a)}_t\geq b_M\}$ and $\kappa(z)=\|b\|_\infty-b(z)$ for all $z\in\Sbb\times \R_+$. The semi-group $R^M$ is the semi-group of the process $Y^M$, defined as follows : it evolves as $Z$ but with an additional killing rate $\kappa$ and killed when its age reaches $\tau_M$. By \textit{killed}, we mean as usual that the process is sent to a cemetery point $\d\notin\Sbb\times\R_+$ at the killing time, in a c\`adl\`ag way. We denote by $\tau^M_\d:=\inf\{t\geq 0,\ Y^M_t\in \d\}$ its killing time, and by $\E^M$ (resp. $\P^M$) the expectation (resp. the probability) associated to the law of $Y^M$, so that
  \[
  \delta_{(s_0,a_0)} R_t^M =\E^M_{s_0,a_0}\left[\ind_{Y^M_t\in
  	\cdot}\ind_{t<\tau^M_\d}\right],\ \forall (s_0,a_0)\in\Sbb\times [0,b_M).
  \]
  
  \begin{lemma}
  	\label{lem:lemM1}
  	There exists positive constants $\lambda^M_0\in(0,\|b\|_\infty)$, $\lambda_1>0$, $C>0$, $t_0\geq 0$, a probability
  	measure $\Upsilon_M$ on $\Sbb\times[0,b_M)$ and a bounded function
  	$\eta:[0,b_M)\rightarrow (0,+\infty)$ such that, for all $t\geq t_0$ and all $(a_0,s_0)\in
  	[0,b_M)\times S$,
  	\[
  	\left\|e^{\lambda_0^M t}\delta_{(s_0,a_0)}R^M_t-\eta(s_0,a_0)\,\Upsilon_M\right\|_{TV}\leq
  	C\,e^{-\lambda_1 t}\,\eta(s_0,a_0).
  	\]
  \end{lemma}

  \begin{proof}
	Our strategy if to check that Assumption~A from~\cite{ChampagnatVillemonais2016b} is satisfied by the process $Y^M$. Once this is done,  Theorem~1.1 in this reference  states that
	\[
			\left\|e^{\lambda^M_0 t}\delta_{(s_0,a_0)}R^M_t-\delta_{(s_0,a_0)}R^M_t\ind_{\Sbb\times[0,b_M)}\,\Upsilon_M\right\|_{TV}\leq
		C\,e^{-\lambda_1 t}\,\delta_{(s_0,a_0)}R^M_t\ind_{\Sbb\times[0,b_M)}.
	\]
	Now, using Theorem~2.1 and Equation~(3.2) in \cite{ChampagnatVillemonais2017b}, we deduce that there exists $t_0>0$ such that, for all $t\geq t_0$,
	\[
		\left|e^{\lambda^M_0 t}\delta_{(s_0,a_0)}R^M_t\ind_{\Sbb\times[0,b_M)}-\eta(s_0,a_0)\right|\leq C\eta(s_0,a_0)e^{-\lambda_1 t},
	\]
	up to a change in the constant $C$, and for some positive bounded function $\eta:\Sbb\times[0,b_M)\to\R_+$. The last two equations allow to conclude the proof of  Lemma~\ref{lem:lemM1} (the fact that $0<\lambda_0^M<\|b\|_\infty$ is a consequence of the fact that $\kappa\leq\|b\|_\infty$ and $0<\Upsilon_M(\kappa)<\|b\|_\infty$).
	
	\medskip
	It remains to check Assumption~A, which is stated as follows.
	
	\noindent\textbf{Assumption A.} There exists a probability measure $\beta$ on $\Sbb\times[0,b_M)$ and positive constants $t_A$, $c_A,c'_A$, such that
	\begin{itemize}
		\item[A1.] for all $x\in\Sbb\times[0,b_M)$,
		\[
			\P^M\left(Y^M_{t_A}\in\cdot\,\mid\,t_A<\tau^M_\d\right)\geq c_A\beta,
		\]
		\item[A2.] for all $x\in\Sbb\times[0,b_M)$ and all $t\geq 0$,
		\[
			\P^M_\beta(t<\tau^M_\d)\geq c'_A \P^M_x(t<\tau^M_\d).
		\]
	\end{itemize}
	The end of the proof, divided in two steps, is dedicated to checking A1 and A2 respectively.
	
	\medskip\noindent\textit{Step 1. Checking A1.} We denote by $\tau_1\leq \cdots\leq \tau_n\leq\cdots$ the successive jump times of $Y^M$, with $\tau_n=+\infty$ if there are less than $n$ jumps. Using the definition of $Y^M$ and the strong Markov property at time $\tau_n$, we obtain that, for all $s_0\in \Sbb$ and all $n\geq 0$,
	\begin{align*}
		\P^M_{(s_0,0)}&\left(Y^M_{\tau_{n+1}}\in \dd s_{n+1}\times\{0\},\tau_{n+1}\in \dd a_{n+1}\right)\\
			&=\int_{\Sbb\times \R_+} \P^M_{(s_0,0)}\left(Y^M_{\tau_{n}}\in \dd s_n\times\{0\},\tau_{n}\in \dd a_n\right) \, \P^M_{(s_n,0)}\left(Y^M_{\tau_1}\in \dd s_{n+1},\ \tau_1\in\dd a_{n+1}-a_n\right)\\
			&\geq \int_{\Sbb\times \R_+} \P^M_{(s_0,0)}\left(Y^M_{\tau_{n}}\in \dd s_n\times\{0\},\tau_{n}\in \dd a_n\right)\,\ind_{a_n\leq a_{n+1}\leq a_n+b_M} e^{-2\|b\|_\infty b_M} \\
			&\qquad\qquad\times\,b(s_n,a_{n+1}-a_n)\,g_{s_n,a_{n+1}-a_n}(s_{n+1})\pi(\dd s_{n+1})\,\lambda(\dd a_{n+1})\\
			&\geq \int_{\Sbb\times \R_+} \P^M_{(s_0,0)}\left(Y^M_{\tau_{n}}\in \dd s_n\times\{0\},\tau_{n}\in \dd a_n\right) \,\uls{b}\,\uls{g}\,\pi(\dd s_{n+1})\,\lambda_{\vert[a_n+b_m,a_n+b_M]}(\dd a_{n+1}),
	\end{align*}
	where $\lambda$ denotes the Lebesgue measure on $\R$, and where $\uls{b}:=e^{-2\|b\|_\infty b_M}\inf_{s\in\Sbb,a\in [b_m,b_M]} b(s,a)$ and $\uls{g}:=\inf_{s_0,s\in\Sbb,a\in[b_m,b_M]}g_{s_0,a}(s)$  are positive (by continuity of $b$ and $g$ and by the compactness of $\Sbb$). Using an iterative procedure, we deduce that, for all $n\geq 1$,
	\begin{align*}
	\P^M_{(s_0,0)}\left(Y^M_{\tau_{n}}\in \dd s_{n}\times\{0\},\tau_{n}\in \dd a_{n}\right)
	&\geq \uls{b}^{n}\,\uls{g}^{n}\,\pi(\dd s_{n})\,\big(\lambda_{\vert [b_m,b_M]}\big)^{\otimes n}(\dd a_{n}).
	\end{align*}
	But $\big(\lambda_{\vert [b_m,b_M]}\big)^{\otimes n}$ admits a positive continuous density on $(nb_m,nb_M)$, hence there exist $n\geq 1$, $c>0$ and $d>b_M$ such that
	\[
		\big(\lambda_{\vert [b_m,b_M]}\big)^{\otimes n}(\dd a_{n})\geq c\lambda_{\vert [d-b_M,d+b_m]}(\dd a_n),
	\]
	so that
	\begin{equation}
	\label{eq:eqM1}
		\P^M_{(s_0,0)}\left(Y^M_{\tau_{n}}\in \dd s_{n}\times\{0\},\tau_{n}\in \dd a_{n}\right)
		\geq c\,\uls{b}^{n}\,\uls{g}^{n}\,\pi(\dd s_{n})\,\lambda_{\vert [d-b_M,d+b_m]}(\dd a_n).
	\end{equation}
	
	Now, let $s_0\in\Sbb$ and $a_0\in(0,b_M)$. We have, for all measurable $A\subset \Sbb$ and all $\ell<\ell'\in[0,b_m]$,
	\begin{align*}
	\P^M_{(s_0,a_0)}&\left(Y^M_{d+b_m}\in A\times [\ell,\ell']\right)
	\geq \P^M_{(s_0,a_0)}\left(Y^M_{d+b_m}\in A\times[\ell,\ell'],\,\tau_{n+1}\in [d,d+b_m]\right)\\
	&\quad\quad\geq \P^M_{(s_0,a_0)}\left(Y^M_{\tau_{n+1}}\in A\times\{0\},\,d+b_m-\tau_{n+1}\in [0,b_m ]\cap[\ell,\ell'],\,\tau_{n+2}> \tau_{n+1}+b_m\right)\\
	&\quad\quad\geq \P^M_{(s_0,a_0)}\left(Y^M_{\tau_{n+1}}\in A\times\{0\},\,d+b_m-\tau_{n+1}\in [\ell,\ell']\right)\,e^{-2\|b\|_\infty b_m},
	\end{align*}
	where we used the strong Markov property at time $\tau_{n+1}$ and the fact that the total jump rate of $Y^M$ (including the killing rate) is uniformly bounded by $2\|b\|_\infty$. Using the strong Markov inequality at time $\tau_1$, we deduce that, for all measurable $A\subset \Sbb$ and all $\ell<\ell'\in[0,b_m)$,
	\begin{align*}
	\P^M_{(s_0,a_0)}\left(Y^M_{d+b_m}\in A\times [\ell,\ell']\right)
	&\geq e^{-2\|b\|_\infty b_m}\,\int_{[a_0,b_M]}\lambda(\dd a_1)\,e^{-2\|b\|_\infty b_M}\,b(s_0,a_1)\int_\Sbb \gamma_{s_0,a_1}(\dd s_1)\\
	&\qquad\qquad\times\P^M_{(s_1,0)}\left(Y^M_{\tau_{n}}\in A\times\{0\},\,d+b_m-\tau_{n}-a_1\in  [\ell,\ell']\right)\\
	&\geq  e^{-4\|b\|_\infty b_M}\,\int_{[a_0,b_M]}\lambda(\dd a_1)\,b(s_0,a_1)\int_\Sbb \gamma_{s_0,a_1}(\dd s_1)\,c\,\uls{b}^{n}\,\uls{g}^{n}\,\pi(A)\\
	&\qquad\qquad\times  \lambda_{\vert [d-2b_M,d+b_m]}([d+b_m-a_1+\ell,d+b_m-a_1+\ell'])\\
	&=e^{-4\|b\|_\infty b_M}\,\int_{[a_0,b_M]}\lambda(\dd a_1)\,b(s_0,a_1)\,c\,\uls{b}^{n}\,\uls{g}^{n}\,\pi(A)\,\lambda([l,l']),
	\end{align*}
	where we used~\eqref{eq:eqM1} for the second inequality. Since $\{d+b_m<\tau_\d\}\subset \{\tau_1\in[a_0,b_M]\}$, we also have
	\[
		\P^M_{(s_0,a_0)}\left(d+b_m<\tau_\d\right)\leq \int_{[a_0,b_M]}\lambda(\dd a_1)\,b(s_0,a_1),
	\]
	so that Assumption~A1 holds true with $t_A=d+b_m$, $c_A=e^{-4\|b\|_\infty b_M}\,c\,\uls{b}^{n}\,\uls{g}^{n}\,b_m$, and
	\begin{equation*}
	\beta(\dd s\times \dd a)=\pi(\dd s)\,\lambda_{\vert [0,b_m]}(\dd a)/b_m.
	\end{equation*}

	\medskip\noindent\textit{Step 2. Checking A2.} On the one hand, for all $(s_0,a_0)\in\Sbb\times[0,b_M)$ and all $t\geq b_m+b_M$, we obtain, using the strong Markov property at time $\tau_1$,
	\begin{align}
	\P_{(s_0,a_0)}\left(t<\tau_\d\right)&\leq \overline{g}\int_\Sbb\,\pi(\dd s)\,\P_{(s,0)}\left(t-a_0<\tau_\d\right)\nonumber\\
	&\leq \overline{g}\int_\Sbb\,\pi(\dd s)\,\int_{[0,b_m]}\frac{\dd a}{b_m}\,\P_{(s,a)}\left(t-a_0-a<\tau_\d\right)\nonumber\\
	&\leq \overline{g}\int_\Sbb\,\pi(\dd s)\,\int_{[0,b_m]}\frac{\dd a}{b_m}\,\P_{(s,a)}\left(t-b_M-b_m<\tau_\d\right),\label{eq:eqM2}
	\end{align}
	where we used the fact that $\P_{(s,a)}\left(u<\tau_\d\right)$ decreases with $u$, and where $\overline{g}:=\sup_{(s_0,a,s)\in\Sbb\times[0,b_M]\times \Sbb}$, which is finite by continuity of $g$ and compactness of $\Sbb$. On the other hand, using Step~1 (where we can and do assume without loss of generality that $t_A\geq b_m+b_M$),
	\begin{align*}
	\P_\beta\left(X_{t_A}\in\cdot\right)\geq \P_\beta(t_A<\tau_\d)\,c_A\,\beta,
	\end{align*}
	so that, using the Markov property at time $t_A$,
	\begin{align}
	\P_\beta\left(t<\tau_\d\right)\geq \P_\beta(t_A<\tau_\d)\,c_A\,\P_\beta\left(t-t_A<\tau_\d\right)\geq \P_\beta(t_A<\tau_\d)\,c_A\,\P_\beta\left(t-b_m-b_M<\tau_\d\right).\label{eq:eqM3}
	\end{align}
	Since $\P_\beta(t_A<\tau_\d)>0$, we deduce from~\eqref{eq:eqM2} and~\eqref{eq:eqM3} that Assumption A2 holds true with 
	\[
		c'_A=\frac{\P_\beta(t_A<\tau_\d)c_A}{\overline{g}}.
	\]
	
	\medskip This conludes the proof of Lemma~\ref{lem:lemM1}.

  \end{proof}

  We introduce now the semi-group $R^\infty$ defined by
  \[
  		\delta_{(s_0,a_0)} R_t^\infty := e^{-\|b\|_\infty t}\delta_{(s_0,a_0)}R_t=\E^Z_{a_0,s_0}\left[\exp\left(\int_0^t
  	-\kappa(Z_u)\,du\right)\ind_{Z_t\in
  		\cdot}\right],
  \]
  which is the semi-group of the Markov process $Y^\infty$ defined as follows : it evolves as $Z$ but with an additional killing $\kappa$ (without killing when at time $\tau_M$, contrarily to $Y^M$). Then we have, denoting by $\E^\infty$ the expectation associated to the law of $Y^\infty$, for all $(s_0,a_0)\in\Sbb\times[0,b_M)$ and all bounded measurable function $f:\Sbb\times\R_+\cup\{\d\}\to\R$ such that $f(\d)=0$,
  \begin{align}
  \label{eq:eqM4}
  	\E^\infty_{(s_0,a_0)}\left(f(Y^\infty_t)\ind_{t<\tau_\d}\right)&=\E^M_{(s_0,a_0)}\left(f(Y^M_t)\ind_{t<\tau_\d}\right)+\E^\infty_{(s_0,a_0)}\left(f(Y^\infty_t)\ind_{\tau_M\leq t<\tau_\d}\right).
  \end{align}

  We will need the following technical result to exhibit the limiting behavior of  $R_t^\infty$, when $t\to+\infty$. Our strategy can be used in general when a reducible process satisfies Assumption~A in a given communication class, and can go into another set where the killing rate is strictly larger than the parameter $\lambda_0^M$ associated to the process restricted to the initial communication class.
  
  \begin{lemma}
  	\label{lem:M2}
  	There exists a constant $C>0$ such that, for all $t\geq 0$ and all $(s,a)\in\Sbb\times\R_+$,
  	\begin{align*}
  		\P^\infty_{(s,a)}\left(\tau_M\leq t<\tau_\d\right)\leq C\,e^{-\lambda^M_0 t}.
  	\end{align*}
  \end{lemma}

\begin{proof}
	We have, using the Markov property at time $\tau_M$,
	\begin{align*}
		\P^\infty_{(s,a)}\left(\tau_M\leq u<\tau_\d\right)&\leq \int_0^t \P^M_{(s,a)}(\tau_\d\in\dd u)e^{-\|b\|_\infty (t-u)}\\
		&\leq \|b\|_\infty e^{-\|b\|_\infty t}\int_0^t \dd v \,\P^M_{(s,a)}(v<\tau_\d)\,e^{\|b\|_\infty v}\\
		&\leq \|b\|_\infty e^{-\|b\|_\infty t}\int_0^t \dd v \,c\,e^{-\lambda_0^M v}\,e^{\|b\|_\infty v},
	\end{align*}
	where $c>0$ is a constant (see Equation~(2.4) of~\cite{ChampagnatVillemonais2016b}). The computation of the right hand term concludes the proof.
\end{proof}

We denote by $\Upsilon_{\text{exit}}$ the law of $Y^M_{\tau^M_\d-}$ under $\P^M_{\Upsilon_M}$. 
\begin{lemma}
	\label{lem:M3}
	There exist positive constants $C,\lambda''$ such that, for all bounded measurable function $f$ and all 
	\begin{equation*}
	\left|e^{\lambda^M_0 t}\E^\infty_{\Upsilon_M}\left(f(Y^\infty_t)\ind_{\tau_M\leq t<\tau_\d}\right)- \lambda_0\int_0^\infty\dd u \,e^{\lambda^M_0 u} \int_{\Sbb\times\{b_M\}} \Upsilon_{\text{exit}}({\rm d}s,{\rm d}a)\E^\infty_{(s,a)}\left(f(Y^\infty_u)\ind_{u<\tau_\d}\right) \right|
	\leq \|f\|_\infty Ce^{-\lambda'' t}.
	\end{equation*}
\end{lemma}
  
  \begin{proof}
  	Under $\P^M_{\Upsilon_M}$, $\tau_\d^M$ is independent from $Y^M_{\tau_\d^M-}$ and is an exponential random variable with parameter $\lambda_0^M$ (these are well known results from the theory of quasi-stationary distributions, see for instance~\cite{ColletMartinezEtAl2013}). Hence, using the strong Markov property at time $\tau_M$ for $Y^\infty$ and using the facts that $Y^\infty_{\tau_M}=Y^\infty_{\tau_M-}$ and that, up to time $\tau_M\wedge \tau_\d$ (excluded), $Y^M$ and $Y^\infty$ have the same law, we obtain 
  	\[
  	\E^\infty_{\Upsilon_M}\left(f(Y^\infty_t)\ind_{\tau_M\leq t<\tau_\d}\right)=\E^M_{\Upsilon_M}\left(\ind_{\tau_\d^M\leq t,\ Y^M_{\tau_\d^M-}\in\Sbb\times\{b_M\}}\E^\infty_{Y^M_{\tau_\d^M-}}\left(f(Y_{t-v})\ind_{t-v<\tau_\d}\right)_{\vert v=\tau_M}\right).
  	\]
  	Then
  	\begin{align*}
  		e^{\lambda_0^M t}\E^\infty_{\Upsilon_M}\left(f(Y^\infty_t)\ind_{\tau_M\leq t<\tau_\d}\right)
  		&=\int_0^t \dd v \,\lambda^M_0 e^{\lambda_0^M(t-v)}\int_{\Sbb\times\{b_M\}} \Upsilon_{\text{exit}}(\dd s,\dd a)\E^\infty_{(s,a)}\left(f(Y^\infty_{t-v})\ind_{t-v<\tau_\d}\right)\\
  		&=\lambda_0 \int_0^t \dd u \,\lambda^M_0 e^{\lambda_0^M u}\int_{\Sbb\times\{b_M\}} \Upsilon_{\text{exit}}(\dd s,\dd a)\E^\infty_{(s,a)}\left(f(Y^\infty_u)\ind_{u<\tau_\d}\right).
  	\end{align*}
  	Now, since $a=b_M$ entails that $\E^\infty_{(s,a)}\left(f(Y^\infty_u)\ind_{u<\tau_\d}\right)\leq \|f\|_\infty e^{-\|b\|_\infty u}$ with $\|b\|_\infty>\lambda_0^M$, one obtains
  	\begin{align*}
  	\int_t^{+\infty} \dd u \,\lambda^M_0 e^{\lambda_0^M u}\int_{\Sbb\times\{b_M\}} \Upsilon_{\text{exit}}(\dd s,\dd a)\E^\infty_{(s,a)}\left(f(Y_u)\ind_{u<\tau_\d}\right)\leq \|f\|_\infty \frac{e^{-(\|b\|_\infty-\lambda_0^M)t}}{\|b\|_\infty-\lambda_0^M},
  	\end{align*}
  	which concludes the proof.  	
  \end{proof}

  The first term on the right hand side of~\eqref{eq:eqM4} multiplied by $e^{\lambda_0^M t}$ converges, according to Lemma~\ref{lem:lemM1}. Let us focus on the second term. We fix $\epsilon>0$ such that $(1-\epsilon)\|b\|_\infty>\lambda_0^M$ and obtain
  \begin{align}
  \E^\infty_{(s_0,a_0)}\left(f(Y_t)\ind_{\tau_M\leq t<\tau_\d}\right)&=\E^\infty_{(s_0,a_0)}\left[\ind_{\epsilon t<\tau_M\leq t}\E^\infty_{Y_{\tau_M}}\left(f(Y_{t-u})\ind_{t-u<\tau_\d}\right)_{\vert u=\tau_M}\right]\label{eq:eqM5}\\
  &\quad\quad\quad+\E^\infty_{(s_0,a_0)}\left(\ind_{\tau_M\leq\, \epsilon t}\E^\infty_{Y_{\tau_M}}\left(f(Y_{t-u})\ind_{t-u<\tau_\d}\right)_{\vert u=\tau_M}\right).\label{eq:eqM6}
  \end{align}
  On the one hand, the term~\eqref{eq:eqM6} is bounded by $\|f\|_\infty e^{-\|b\|_\infty (1-\epsilon) t}$. On the other hand, setting $g(y,u)=\E_{y}\left(f(Y_{t-u})\ind_{t-u<\tau_\d}\right)$, we have
  \begin{align*}
  \E^\infty_{(s_0,a_0)}\left[\ind_{\epsilon t<\tau_M\leq t}g(Y_{\tau_M},\tau_{M})\right]
  &=\E^M_{(s_0,a_0)}\left[\ind_{\epsilon t<\tau_M}\E^\infty_{Y_{\epsilon t}}\left(g(Y_{\tau_M},\tau_{M}+\epsilon t)\ind_{\tau_M\leq (1-\epsilon)t}\right)\right]\\
  &=e^{-\epsilon t\lambda_0^M}\eta(s_0,a_0)\E^\infty_{\Upsilon_M}\left(g(Y_{\tau_M},\tau_{M}+\epsilon t)\ind_{\tau_M\leq (1-\epsilon)t}\right)\\
  &\quad\quad\quad +{\cal O}(e^{-(\lambda^M_0+\lambda_1) \epsilon t})\sup_{(s,a)\in\Sbb\times\R_+}\E^\infty_{(s,a)}\left(g(Y_{\tau_M},\tau_{M}+\epsilon t)\ind_{\tau_M\leq (1-\epsilon)t}\right),
  \end{align*}
  where ${\cal O}(e^{-\lambda_1 t})$ is uniform in $(s_0,a_0)$ by Lemma~\ref{lem:lemM1}. We note that, according to Lemma~\ref{lem:M2},
  \begin{align*}
  \E^\infty_{(s,a)}\left(g(Y_{\tau_M},\tau_{M}+\epsilon t)\ind_{\tau_M\leq (1-\epsilon)t}\right)&\leq \|f\|_\infty \P^\infty_{(s,a)}\left(\tau_M\leq (1-\epsilon)t<\tau_\d\right)\leq C\|f\|_\infty e^{-(1-\epsilon)\lambda^M_0 t}.
  \end{align*}
  As a consequence (using also the bound on~\eqref{eq:eqM6}), there exists a constant $\lambda'>0$ such that
  \begin{align}
  \E^\infty_{(s_0,a_0)}\left(f(Y_t)\ind_{\tau_M\leq t<\tau_\d}\right)&= e^{-\epsilon t\lambda_0^M}\eta(s_0,a_0)\E^\infty_{\Upsilon_M}\left(g(Y_{\tau_M},\tau_{M}+\epsilon t)\ind_{\tau_M\leq (1-\epsilon)t}\right)+{\cal O}(e^{-(\lambda_0^M+\lambda')t}),\\
  &=e^{-\epsilon t\lambda_0^M}\eta(s_0,a_0)\E^M_{\Upsilon_M}\left(g(Y_{\tau_M-},\tau_{M}+\epsilon t)\ind_{\tau_M\leq (1-\epsilon)t}\right)+{\cal O}(e^{-(\lambda_0^M+\lambda')t}),\label{eq:eqM7}
  \end{align}
  uniformly in $(s_0,a_0)\in\Sbb\times[0,b_M)$. But the same procedure as in the proof of Lemma~\ref{lem:M3} shows that
  \begin{align*}
  \E^M_{\Upsilon_M}\left(g(Y_{\tau_M-},\tau_{M}+\epsilon t)\ind_{\tau_M\leq (1-\epsilon)t}\right)&=\int_0^{(1-\epsilon)t} \dd u \lambda_0^M\,e^{-\lambda_0^M u} \int_{\Sbb\times \{b_M\}} \Upsilon_{\text{exit}}(\dd y) g(y,u+\epsilon t)\\
  &=e^{\epsilon \lambda_0^Mt}\int_{\epsilon t}^{t} \dd v \lambda_0^M\,e^{-\lambda_0^M v} \int_{\Sbb\times \{b_M\}} \Upsilon_{\text{exit}}(\dd y) g(y,v).
  \end{align*}
  Since, $ \Upsilon_{\text{exit}}(\dd y)$-almost surely, $g(y,v)$ is bounded by $\|f\|_\infty e^{-\|b\|_\infty (t-v)}$, we have
  \begin{align*}
 & \left|e^{-\epsilon\lambda_0^M t}\E^M_{\Upsilon_M}\left(g(Y_{\tau_M-},\tau_{M}+\epsilon t)\ind_{\tau_M\leq (1-\epsilon)t}\right)- \E^M_{\Upsilon_0}\left(g(Y_{\tau_M},\tau_{M})\ind_{\tau_M\leq t}\right)\right|\\
 &\phantom{e^{-\epsilon\lambda_0^M t}\E^\infty_{\Upsilon_M}\left(g(Y_{\tau_M},\tau_{M}+\epsilon t)\ind_{\tau_M\leq (1-\epsilon)t}\right)- \E^\infty_{\Upsilon_M}}
 \leq \|f\|_\infty \int_{0}^{\epsilon t} \dd v \lambda_0^M\,e^{-\lambda_0^M v} e^{-\|b\|_\infty (t-v)}\\
 &\phantom{e^{-\epsilon\lambda_0^M t}\E^\infty_{\Upsilon_0}\left(g(Y_{\tau_M^c},\tau_{M^c}+\epsilon t)\ind_{\tau_M^c\leq (1-\epsilon)t}\right)- \E^\infty_{\Upsilon_0}}
 =\|f\|_\infty \lambda_0^M e^{-\|b\|_\infty t}\frac{e^{\epsilon(\|b\|_\infty-\lambda_0^M)t}-1}{\|b\|_\infty-\lambda_0^M}\\
 &\phantom{e^{-\epsilon\lambda_0^M t}\E^\infty_{\Upsilon_0}\left(g(Y_{\tau_M^c},\tau_{M^c}+\epsilon t)\ind_{\tau_M^c\leq (1-\epsilon)t}\right)- \E^\infty_{\Upsilon_0}}
 =e^{-\lambda_0^M t}\,{\cal O}(e^{-(1-\epsilon)(\|b\|_\infty-\lambda_0^M)t}).
  \end{align*}
  Using the last inequality, combined with~\eqref{eq:eqM7}, $\eqref{eq:eqM4}$ and Lemma~\ref{lem:M3}, we deduce that
  \begin{align*}
  \left|e^{\lambda_0^M t}\E^\infty_{(s_0,a_0)}\left(f(Y_t)\ind_{t<\tau_\d}\right)-\eta(s_0,a_0)\Upsilon(f)\right|\leq C\,\|f\|_\infty\,e^{-\lambda t},
  \end{align*}
  for some positive constants $C>0$ and $\lambda>0$, where
  \begin{equation}
  \label{eq:Qproc}
	  \Upsilon(f)=\Upsilon_M(f)+\lambda_0\int_0^\infty\dd u \,e^{\lambda^M_0 u} \int_{\Sbb\times\{b_M\}} \Upsilon_{\text{exit}}(\dd s,\dd a)\,\E^\infty_{(s,a)}\left(f(Y_u)\ind_{u<\tau_\d}\right).
  \end{equation}

  The previous analysis was valid for $(s_0,a_0)\in\Sbb\times[0,b_M)$. When $(s_0,a_0)\in\Sbb\times[b_M,+\infty)$, then the killing rate of the process is $\|b\|_\infty$, so that the last inequality holds true (up to a modification of $C$ and $\lambda$) with $\eta(s_0,a_0)=0$.
  
  Taking $\lambda_0=\lambda^M_0-\|b\|_\infty$, this concludes the proof of the first part of Theorem~\ref{thm:malthus1}. 
  
  \medskip
  
  Let us now prove the last assertion of the theorem. Fix $\lambda >-\lambda_0$. From Corollary~\ref{cor:useful}, we now that, for all $s_0\in\Sbb$ and all $t\geq 0$,
  \[
  \nu_{\lambda}(s_0,\Sbb\times[0,t])=e^{-\lambda t}\E_{s_0,0}^Z\left(e^{\lambda Z_t^{(a)}}\,e^{\int_0^t b(Z_u)\,\mathrm du}\right)=e^{(\|b\|_\infty-\lambda) t}\E^\infty_{(s_0,0)}\left(e^{\lambda Y_t^{(a)}}\ind_{t<\tau_\d}\right)
  \]
  where $Z$ is the process described in Section~\ref{sec:poissonian}. Note that $Y_t^{(a)}\leq b_M$ for all $t\leq \tau_M$ and since $Y_t^{(a)}=t-\tau_M$ for all $t\in[\tau_M,\tau_\d)$. Hence, ne the one hand,
  \begin{align*}
  \E^\infty_{(s_0,0)}\left(e^{\lambda Y_t^{(a)}}\ind_{t<\tau_\d\wedge\tau_M}\right)&\leq e^{\lambda b_M}\P^\infty_{s_0,0}\left(t<\tau_\d\wedge\tau_M\right).
  \end{align*}
  But $\|b\|_\infty-\lambda<\|b\|_\infty+\lambda_0=\lambda_0^M$, so that, according to~\eqref{eq:Qproc}, 
  \begin{equation}
  \label{eq:rhs1}
  e^{(\|b\|_\infty-\lambda) t}\E^\infty_{(s_0,0)}\left(e^{\lambda Y_t^{(a)}}\ind_{t<\tau_\d\wedge\tau_M}\right)\xrightarrow[t\to+\infty]{}0.
  \end{equation}
  On the other hand, we have
  \begin{align*}
  \E^\infty_{(s_0,0)}\left(e^{\lambda Y_t^{(a)}}\ind_{\tau_M\leq t<\tau_\d}\right)&= e^{\lambda t}\E^\infty_{(s_0,0)}\left(e^{-\lambda \tau_M}\ind_{\tau_M\leq t<\tau_\d}\right)\\
  &\leq e^{\lambda t}\E^M_{(s_0,0)}\left(e^{-\lambda \tau_M}\ind_{\tau_M=\tau_\d}\right)\\
  &\leq e^{\lambda t}\E^M_{(s_0,0)}\left(e^{-\lambda \tau_\d}\right).
  \end{align*}
  But the killing time $\tau_\d$ under $\P^M_{(s_0,a_0)}$has an exponential queue with parameter $\lambda_0^M$ (see for instance Proposition~2.3 in~\cite{ChampagnatVillemonais2016b}), so that there exists a constant $C>0$ such that
  \begin{align*}
  \E^\infty_{(s_0,0)}\left(e^{\lambda Y_t^{(a)}}\ind_{\tau_M\leq t<\tau_\d}\right)&\leq Ce^{\lambda t}\int_0^t e^{-\lambda u}e^{-\lambda^M_0 u}\,\mathrm du=C\exp(\lambda t)\frac{e^{-(\lambda+\lambda_0^M)t}-1}{\lambda+\lambda_0^M}
  \end{align*}
  Hence, using the fact that $\|b\|_\infty-\lambda_0^M=-\lambda_0<\lambda$,
  \begin{align*}
  \E^\infty_{(s_0,0)}\left(e^{\lambda Y_t^{(a)}}\ind_{\tau_M\leq t<\tau_\d}\right)&\leq C \frac{e^{(\|b\|_\infty-\lambda-\lambda_0^M)t}}{\lambda+\lambda_0^M}\xrightarrow[t\to+\infty]{}0.
  \end{align*}
  Finally, we have proved that, for all $\lambda >-\lambda_0$,
  \begin{align*}
  \nu_\lambda(s_0,\Sbb\times\R_+)=\lim_{t\to+\infty} \nu_\lambda(s_0,\Sbb\times[0,t])=0,
  \end{align*}
  so that $\alpha\leq -\lambda_0$.
  
  Finally, one observes that, according to the already proved first part of Theorem~\ref{thm:malthus1},
  \begin{align*}
  \nu_{-\lambda_0}(s_0,\Sbb\times[0,t])\geq e^{\lambda_0 t} \delta_{s_0,0} R_t \xrightarrow[t\to+\infty]{} \eta(s_0,a_0)\Upsilon(\ind_{\Sbb\times\R_+})>0,
  \end{align*}
  so that $\alpha\geq -\lambda_0$. This concludes the proof of Theorem~\ref{thm:malthus1}.
  

\begin{thebibliography}{50}
    \providecommand{\natexlab}[1]{#1}
    \providecommand{\url}[1]{\texttt{#1}}
    \expandafter\ifx\csname urlstyle\endcsname\relax
    \providecommand{\doi}[1]{doi: #1}\else
    \providecommand{\doi}{doi: \begingroup \urlstyle{rm}\Url}\fi
    
    \bibitem[Aviv and Susser(2013)]{AvivSusser2013}
    Abraham Aviv and Ezra Susser (2013)
    \newblock Leukocyte telomere length and the father’s age enigma: implications
    for population health and for life course.
    \newblock \emph{International journal of epidemiology}, 42\penalty0
    (2):\penalty0 457--462.
    
    \bibitem[Aza\"\i~s et~al.(2014)Aza\"\i~s, Bardet, G\'enadot, Krell, and
    Zitt]{AzaisBardetEtAl2014}
    Romain Aza\"\i~s, Jean-Baptiste Bardet, Alexandre G\'enadot, Nathalie Krell,
    and Pierre-Andr\'e Zitt (2014)
    \newblock Piecewise deterministic {M}arkov process---recent results.
    \newblock In \emph{Journ\'ees {MAS} 2012}, volume~44 of \emph{ESAIM Proc.},
    pages 276--290. EDP Sci., Les Ulis.
    \newblock URL \url{https://doi.org/10.1051/proc/201444017}.
    
    \bibitem[Bansaye(2009)]{Bansaye2009}
    Vincent Bansaye (2009)
    \newblock Surviving particles for subcritical branching processes in random
    environment.
    \newblock \emph{Stochastic Process. Appl.}, 119\penalty0 (8):\penalty0
    2436--2464.
    \newblock ISSN 0304-4149.
    \newblock \doi{10.1016/j.spa.2008.12.005}.
    \newblock URL \url{http://dx.doi.org/10.1016/j.spa.2008.12.005}.
    
    \bibitem[Bansaye and M{\'e}l{\'e}ard(2015)]{BansayeMeleard2015}
    Vincent Bansaye and Sylvie M{\'e}l{\'e}ard (2015)
    \newblock \emph{Stochastic models for structured populations}, volume~16.
    \newblock Springer.
    
    \bibitem[Benetos et~al.(2018)Benetos, Toupance, Gautier, Labat, Kimura, Rossi,
    Settembre, Hubert, Frimat, Bertrand, et~al.]{BenetosToupanceEtAl2018}
    Athanase Benetos, Simon Toupance, Sylvie Gautier, Carlos Labat, Masayuki
    Kimura, Pascal~M Rossi, Nicla Settembre, Jacques Hubert, Luc Frimat, Baptiste
    Bertrand, et~al (2018)
    \newblock Short leukocyte telomere length precedes clinical expression of
    atherosclerosis: the blood-and-muscle model.
    \newblock \emph{Circulation research}, 122\penalty0 (4):\penalty0 616--623.
    
    \bibitem[Bertoin(2017)]{Bertoin2017}
    Jean Bertoin (2017)
    \newblock Markovian growth-fragmentation processes.
    \newblock \emph{Bernoulli}, 23\penalty0 (2):\penalty0 1082--1101.
    
    \bibitem[Bourgeron et~al.(2015)Bourgeron, Xu, Doumic, and
    Teixeira]{BourgeronXuEtAl2015a}
    Thibault Bourgeron, Zhou Xu, Marie Doumic, and Maria~Teresa Teixeira (2015)
    \newblock The asymmetry of telomere replication contributes to replicative
    senescence heterogeneity.
    \newblock \emph{Scientific reports}, 5\penalty0 (1):\penalty0 1--11.
    
    \bibitem[Broer et~al.(2013)]{BroerEtAl2013}
        Linda Broer, Veryan Codd, Dale R Nyholt, Joris Deelen, Massimo Mangino, Gonneke Willemsen, Eva Albrecht, Najaf Amin, Marian Beekman, Eco J C de Geus, Anjali Henders, Christopher P Nelson, Claire J Steves, Margie J Wright, Anton J M de Craen, Aaron Isaacs, Mary Matthews, Alireza Moayyeri, Grant W Montgomery, Ben A Oostra, Jacqueline M Vink, Tim D Spector, P Eline Slagboom, Nicholas G Martin, Nilesh J Samani, Cornelia M van Duijn, Dorret I Boomsma (2013)
        \newblock Meta-analysis of telomere length in 19713 subjects reveals high heritability, stronger maternal inheritance and a paternal age effect.
        \newblock \emph{European Journal of Human Genetics}, 21\penalty0 1163--1168.
    
    \bibitem[Buxton et~al.(2011)Buxton, Walters, Visvikis-Siest, Meyre, Froguel,
    and Blakemore]{BuxtonWaltersEtAl2011}
    Jessica~L Buxton, Robin~G Walters, Sophie Visvikis-Siest, David Meyre, Philippe
    Froguel, and Alexandra~IF Blakemore (2011)
    \newblock Childhood obesity is associated with shorter leukocyte telomere
    length.
    \newblock \emph{The Journal of Clinical Endocrinology \& Metabolism},
    96\penalty0 (5):\penalty0 1500--1505.
    
    \bibitem[Champagnat and Villemonais(2016)]{ChampagnatVillemonais2016b}
    Nicolas Champagnat and Denis Villemonais.
    \newblock Exponential convergence to quasi-stationary distribution and
    {Q}-process (2016)
    \newblock \emph{Probab. Theory Related Fields}, 164\penalty0 (1):\penalty0
    243--283.
    \newblock ISSN 1432-2064.
    \newblock \doi{10.1007/s00440-014-0611-7}.
    \newblock URL \url{http://dx.doi.org/10.1007/s00440-014-0611-7}.
    
    \bibitem[Champagnat and Villemonais(2017)]{ChampagnatVillemonais2017b}
    Nicolas Champagnat and Denis Villemonais (2017)
    \newblock Uniform convergence to the $q$-process.
    \newblock \emph{Electron. Commun. Probab.}, 22:\penalty0 7 pp.
    \newblock \doi{10.1214/17-ECP63}.
    \newblock URL \url{https://doi.org/10.1214/17-ECP63}.
    
    \bibitem[Collet et~al.(2013)Collet, Mart\'inez, and
    San~Mart\'in]{ColletMartinezEtAl2013}
    Pierre Collet, Servet Mart\'inez, and Jaime San~Mart\'in (2013)
    \newblock \emph{Quasi-stationary distributions}.
    \newblock Probability and its Applications (New York). Springer, Heidelberg.
    \newblock ISBN 978-3-642-33130-5; 978-3-642-33131-2.
    \newblock \doi{10.1007/978-3-642-33131-2}.
    \newblock URL \url{http://dx.doi.org/10.1007/978-3-642-33131-2}.
    \newblock Markov chains, diffusions and dynamical systems.
    
    \bibitem[Daley(1968)]{Daley1968}
    Daryl~J Daley (1968)
    \newblock Extinction conditions for certain bisexual galton-watson branching
    processes.
    \newblock \emph{Zeitschrift f{\"u}r Wahrscheinlichkeitstheorie und verwandte
        Gebiete}, 9\penalty0 (4):\penalty0 315--322.
    
    \bibitem[Daley et~al.(1986)Daley, Hull, and Taylor]{DaleyHullEtAl1986}
    Daryl~J Daley, David~M Hull, and James~M Taylor (1986)
    \newblock Bisexual galton--watson branching processes with superadditive mating
    functions.
    \newblock \emph{Journal of applied probability}, 23\penalty0 (3):\penalty0
    585--600.
    
    \bibitem[Davis(1984)]{Davis1984}
    M.~H.~A. Davis (1984)
    \newblock Piecewise-deterministic markov processes: A general class of
    non-diffusion stochastic models.
    \newblock \emph{Journal of the Royal Statistical Society: Series B
        (Methodological)}, 46\penalty0 (3):\penalty0 353--376.
    \newblock \doi{10.1111/j.2517-6161.1984.tb01308.x}.
    \newblock URL
    \url{https://rss.onlinelibrary.wiley.com/doi/abs/10.1111/j.2517-6161.1984.tb01308.x}.
    
    \bibitem[Davis(1993)]{Davis1993}
    Mark~HA Davis (1993)
    \newblock \emph{Markov models \& optimization}, volume~49.
    \newblock CRC Press.
    
    \bibitem[De~Meyer et~al.(2007)De~Meyer, Rietzschel, De~Buyzere, De~Bacquer,
    Van~Criekinge, De~Backer, Gillebert, Van~Oostveldt, and
    Bekaert]{DeMeyerRietzschelEtAl2007}
    Tim De~Meyer, Ernst~R Rietzschel, Marc~L De~Buyzere, Dirk De~Bacquer, Wim
    Van~Criekinge, Guy~G De~Backer, Thierry~C Gillebert, Patrick Van~Oostveldt,
    and Sofie Bekaert (2007)
    \newblock Paternal age at birth is an important determinant of offspring
    telomere length.
    \newblock \emph{Human molecular genetics}, 16\penalty0 (24):\penalty0
    3097--3102.
    
    \bibitem[Entringer et~al.(2011)Entringer, Epel, Kumsta, Lin, Hellhammer,
    Blackburn, W{\"u}st, and Wadhwa]{EntringerEpelEtAl2011}
    Sonja Entringer, Elissa~S Epel, Robert Kumsta, Jue Lin, Dirk~H Hellhammer,
    Elizabeth~H Blackburn, Stefan W{\"u}st, and Pathik~D Wadhwa (2011)
    \newblock Stress exposure in intrauterine life is associated with shorter
    telomere length in young adulthood.
    \newblock \emph{Proceedings of the National Academy of Sciences}, 108\penalty0
    (33):\penalty0 E513--E518.
    
    \bibitem[Entringer et~al.(2018)Entringer, de~Punder, Buss, and
    Wadhwa]{EntringerPunderEtAl2018}
    Sonja Entringer, Karin de~Punder, Claudia Buss, and Pathik~D Wadhwa.
    \newblock The fetal programming of telomere biology hypothesis: an update (2018)
    \newblock \emph{Philosophical Transactions of the Royal Society B: Biological
        Sciences}, 373\penalty0 (1741):\penalty0 20170151.
    
    \bibitem[Frenck et~al.(1998)Frenck, Blackburn, and
    Shannon]{FrenckBlackburnEtAl1998}
    Robert~W. Frenck, Elizabeth~H Blackburn, and Kevin~M. Shannon (1998)
    \newblock The rate of telomere sequence loss in human leukocytes varies with
    age.
    \newblock \emph{Proceedings of the National Academy of Sciences of the United
        States of America}, 95 10:\penalty0 5607--10.
    
    \bibitem[Fritsch et~al.(2022)Fritsch, Villemonais, and
    Zalduendo]{FritschVillemonaisEtAl2022}
    Coralie Fritsch, Denis Villemonais, and Nicol{\'a}s Zalduendo (2022)
    \newblock The multi-type bisexual galton-watson branching process.
    \newblock \emph{arXiv preprint arXiv:2206.09622}.
    
    \bibitem[Glei et~al.(2016)Glei, Goldman, Risques, Rehkopf, Dow, Rosero-Bixby,
    and Weinstein]{GleiGoldmanEtAl2016}
    Dana~A Glei, Noreen Goldman, Rosa~Ana Risques, David~H Rehkopf, William~H Dow,
    Luis Rosero-Bixby, and Maxine Weinstein (2016)
    \newblock Predicting survival from telomere length versus conventional
    predictors: a multinational population-based cohort study.
    \newblock \emph{PLoS one}, 11\penalty0 (4).
    
    \bibitem[Gonz{\'a}lez and Molina(1996)]{GonzalezMolina1996}
    Miguel Gonz{\'a}lez and Manuel Molina (1996)
    \newblock On the limit behaviour of a superadditive bisexual galton--watson
    branching process.
    \newblock \emph{Journal of applied probability}, 33\penalty0 (4):\penalty0
    960--967.
    
    \bibitem[Haccou et~al.(2005)Haccou, Haccou, Jagers, Vatutin, and
    Vatutin]{HaccouHaccouEtAl2005}
    Patsy Haccou, Patricia Haccou, Peter Jagers, Vladimir~A Vatutin, and Vladimir
    Vatutin (2005)
    \newblock \emph{Branching processes: variation, growth, and extinction of
        populations}.
    \newblock Number~5. Cambridge university press.
    
    \bibitem[Harris(1963)]{Harris1963}
    Theodore~E. Harris.
    \newblock \emph{The theory of branching processes} (1963)
    \newblock Die Grundlehren der Mathematischen Wissenschaften, Bd. 119.
    Springer-Verlag, Berlin; Prentice-Hall, Inc., Englewood Cliffs, N.J.
    
    \bibitem[Hjelmborg et al.(2015)HJELMBORG, Jacob B., DALGÅRD, Christine, MÖLLER, Soren, et al. ]{HjelmborgEtAl2015}
    Jacob B Hjelmborg, Christine Dalgård, Soren Möller, Troels Steenstrup, Masayuki Kimura, Kaare Christensen, Kirsten O Kyvik, Abraham Aviv (2015)
    \newblock The heritability of leucocyte telomere length
    dynamics
    \newblock \emph{Journal of Medical Genetics}, 52\penalty0 (5):\penalty0 297--302.
    
    \bibitem[Holohan et~al.(2015)Holohan, De~Meyer, Batten, Mangino, Hunt, Bekaert,
    De~Buyzere, Rietzschel, Spector, Wright, et~al.]{HolohanDeMeyerEtAl2015}
    Brody Holohan, Tim De~Meyer, Kimberly Batten, Massimo Mangino, Steven~C Hunt,
    Sofie Bekaert, Marc~L De~Buyzere, Ernst~R Rietzschel, Tim~D Spector,
    Woodring~E Wright, et~al (2015)
    \newblock Decreasing initial telomere length in humans intergenerationally
    understates age-associated telomere shortening.
    \newblock \emph{Aging Cell}, 14\penalty0 (4):\penalty0 669--677.
    
    \bibitem[Honig et~al.(2015)]{HonigEtAl2015}
    Lawrence S. Honig, Min Suk Kang, Rong Cheng, John H. Eckfeldt, Bharat Thyagarajan, Catherine Leiendecker-Foster, Michael A. Province, Jason L. Sanders, Thomas Perls, Kaare Christensen, Joseph H. Lee, Richard Mayeux, Nicole Schupf (2015)
    \newblock Heritability of telomere length in a study of long-lived families.
    \newblock \emph{Neurobiology of aging}, 36\penalty0 (10):\penalty0 2785--2790.
    
    \bibitem[Hull(2003)]{Hull2003}
    David~M Hull (2003)
    \newblock A survey of the literature associated with the bisexual galton-watson
    branching process.
    \newblock \emph{Extracta mathematicae}, 18\penalty0 (3):\penalty0 321--343.
    
    \bibitem[Jagers(1989)]{Jagers1989}
    Peter Jagers (1989)
    \newblock General branching processes as markov fields.
    \newblock \emph{Stochastic Processes and their Applications}, 32\penalty0
    (2):\penalty0 183--212.
    
    \bibitem[Jagers and Nerman(1984)]{JagersNerman1984}
    Peter Jagers and Olle Nerman (1984)
    \newblock The growth and composition of branching populations.
    \newblock \emph{Advances in applied probability}, 16\penalty0 (2):\penalty0
    221--259.
    
    \bibitem[Kaplan(1974)]{Kaplan1974}
    Norman Kaplan (1974)
    \newblock Some results about multidimensional branching processes with random
    environments.
    \newblock \emph{The Annals of Probability}, pages 441--455.
    
    \bibitem[Kawazu and Watanabe(1971)]{KawazuWatanabe1971a}
    Kiyoshi Kawazu and Shinzo Watanabe (1971)
    \newblock Branching processes with immigration and related limit theorems.
    \newblock \emph{Theory of Probability \& Its Applications}, 16\penalty0
    (1):\penalty0 36--54.
    
    \bibitem[Keiding(1976)]{Keiding1976}
    Niels Keiding (1976)
    \newblock Population growth and branching processes in random environments.
    \newblock \emph{Proceedings of the 9th Internatmnul Biometric ConJrrmce}, pages
    149--165.
    
    \bibitem[Kersting and Vatutin(2017)]{KerstingVatutin2017}
    G{\"o}tz Kersting and Vladimir Vatutin (2017)
    \newblock \emph{Discrete time branching processes in random environment}.
    \newblock John Wiley \& Sons.
    
    \bibitem[Laberthonnière et~al.(2019)Laberthonnière, Magdinier, and
    Robin]{LaberthonniereMagdinierEtAl2019}
    Camille Laberthonnière, Frédérique Magdinier, and Jérôme~D. Robin (2019)
    \newblock Bring it to an end: Does telomeres size matter?
    \newblock \emph{Cells}, 8\penalty0 (1).
    \newblock ISSN 2073-4409.
    \newblock \doi{10.3390/cells8010030}.
    \newblock URL \url{https://www.mdpi.com/2073-4409/8/1/30}.
    
    \bibitem[Lee and Kimmel(2020)]{LeeKimmel2020}
    Kyung~Hyun Lee and Marek Kimmel (2020)
    \newblock Stationary distribution of telomere lengths in cells with telomere
    length maintenance and its parametric inference.
    \newblock \emph{Bulletin of Mathematical Biology}, 82\penalty0 (12):\penalty0
    150.
    
    \bibitem[Li(2006)]{Li2006}
    Zeng-hu Li (2006)
    \newblock Branching processes with immigration and related topics.
    \newblock \emph{Frontiers of Mathematics in China}, 1:\penalty0 73--97.
    
    \bibitem[Mattarocci et~al.(2021)Mattarocci, Berardi, Langston, Marcand, Doumic,
    Xu, and Teixeira]{MattarocciBerardiEtAl2021a}
    Stefano Mattarocci, Prisca Berardi, Rachel Langston, St{\'e}phane Marcand,
    Marie Doumic, Zhou Xu, and Maria~Teresa Teixeira (2021)
    \newblock The effect of the shortest telomere on cell proliferation.
    \newblock In \emph{TELOMERES \& TELOMERASE}.
    
    \bibitem[M{\'e}l{\'e}ard and Villemonais(2012)]{MeleardVillemonais2012}
    Sylvie M{\'e}l{\'e}ard and Denis Villemonais (2012)
    \newblock Quasi-stationary distributions and population processes.
    \newblock \emph{Probab. Surv.}, 9:\penalty0 340--410.
    \newblock ISSN 1549-5787.
    
    \bibitem[Molina(2010)]{Molina2010}
    Manuel Molina.
    \newblock Two-sex branching process literature (2010)
    \newblock In \emph{Workshop on branching processes and their applications},
    pages 279--293. Springer.
    
    \bibitem[Nordfj{\"a}ll et~al.(2005)Nordfj{\"a}ll, Larefalk, Lindgren, Holmberg,
    and Roos]{NordfjaellLarefalkEtAl2005}
    Katarina Nordfj{\"a}ll, {\AA}sa Larefalk, Petter Lindgren, Dan Holmberg, and
    G{\"o}ran Roos (2005)
    \newblock Telomere length and heredity: Indications of paternal inheritance.
    \newblock \emph{Proceedings of the National Academy of Sciences}, 102\penalty0
    (45):\penalty0 16374--16378.
    
    \bibitem[Olofsson(2009)]{Olofsson2009}
    Peter Olofsson (2009)
    \newblock Size-biased branching population measures and the multi-type x log x
    condition.
    \newblock \emph{Bernoulli}, 15\penalty0 (4):\penalty0 1287--1304.
    
    \bibitem[Olofsson and Kimmel(1999)]{OlofssonKimmel1999a}
    Peter Olofsson and Marek Kimmel (1999)
    \newblock Stochastic models of telomere shortening.
    \newblock \emph{Mathematical biosciences}, 158\penalty0 (1):\penalty0 75--92.
    
    \bibitem[Pakes(1971)]{Pakes1971}
    AG~Pakes.
    \newblock Branching processes with immigration (1971)
    \newblock \emph{Journal of Applied Probability}, 8\penalty0 (1):\penalty0
    32--42.
    
    \bibitem[Pazy(1983)]{Pazy1983}
    A.~Pazy (1983)
    \newblock \emph{Semigroups of linear operators and applications to partial
        differential equations}, volume~44 of \emph{Applied Mathematical Sciences}.
    \newblock Springer-Verlag, New York.
    \newblock ISBN 0-387-90845-5.
    \newblock \doi{10.1007/978-1-4612-5561-1}.
    \newblock URL \url{http://dx.doi.org/10.1007/978-1-4612-5561-1}.
    
    \bibitem[Shalev et~al.(2013)Shalev, Moffitt, Sugden, Williams, Houts, Danese,
    Mill, Arseneault, and Caspi]{ShalevMoffittEtAl2013}
    Idan Shalev, Terrie~E Moffitt, Karen Sugden, Brittany Williams, Renate~M Houts,
    Andrea Danese, Jonathan Mill, L~Arseneault, and Avshalom Caspi (2013)
    \newblock Exposure to violence during childhood is associated with telomere
    erosion from 5 to 10 years of age: a longitudinal study.
    \newblock \emph{Molecular psychiatry}, 18\penalty0 (5):\penalty0 576--581.
    
    \bibitem[Smith and Wilkinson(1969)]{SmithWilkinson1969}
    Walter~L Smith and William~E Wilkinson (1969)
    \newblock On branching processes in random environments.
    \newblock \emph{The Annals of Mathematical Statistics}, pages 814--827.
    
    \bibitem[van Doorn and Pollett(2013)]{DoornPollett2013}
    Erik~A. van Doorn and Philip~K. Pollett (2013)
    \newblock Quasi-stationary distributions for discrete-state models.
    \newblock \emph{European J. Oper. Res.}, 230\penalty0 (1):\penalty0 1--14.
    \newblock ISSN 0377-2217.
    \newblock \doi{10.1016/j.ejor.2013.01.032}.
    \newblock URL \url{http://dx.doi.org/10.1016/j.ejor.2013.01.032}.
    
    \bibitem[Varadarajan(1958)]{Varadarajan1958}
    V.~S. Varadarajan (1958)
    \newblock Weak convergence of measures on separable metric spaces.
    \newblock \emph{Sankhy\=a}, 19:\penalty0 15--22.
    \newblock ISSN 0972-7671.
    
    \bibitem[Whittemore et~al.(2019)Whittemore, Vera, Mart{\'\i}nez-Nevado,
    Sanpera, and Blasco]{WhittemoreVeraEtAl2019}
    Kurt Whittemore, Elsa Vera, Eva Mart{\'\i}nez-Nevado, Carola Sanpera, and
    Maria~A Blasco (2019)
    \newblock Telomere shortening rate predicts species life span.
    \newblock \emph{Proceedings of the National Academy of Sciences}, 116\penalty0
    (30):\penalty0 15122--15127.
    
    \bibitem[Xu et~al.(2013)Xu, Duc, Holcman, and Teixeira]{XuDucEtAl2013}
    Zhou Xu, Khanh~Dao Duc, David Holcman, and Maria~Teresa Teixeira (2013)
    \newblock The length of the shortest telomere as the major determinant of the
    onset of replicative senescence.
    \newblock \emph{Genetics}, 194\penalty0 (4):\penalty0 847--857.
    
    \bibitem[Zvereva et~al.(2010)Zvereva, Shcherbakova, and
    Dontsova]{ZverevaShcherbakovaEtAl2010}
    MI~Zvereva, DM~Shcherbakova, and OA~Dontsova (2010)
    \newblock Telomerase: structure, functions, and activity regulation.
    \newblock \emph{Biochemistry (Moscow)}, 75\penalty0 (13):\penalty0 1563--1583.
    
\end{thebibliography}

\end{document}